\documentclass[11pt]{article}

\usepackage{fullpage}
\usepackage{amsmath,amsfonts,amsthm,mathrsfs,mathpazo,xspace,hyperref,graphicx}
\usepackage{endnotes}
\usepackage{color}
\usepackage{bbm}
\usepackage{times}
\usepackage{amssymb,latexsym}

\newtheorem*{theorem*}{Theorem}
\newtheorem{theorem}{Theorem}[section]
\newtheorem{proposition}[theorem]{Proposition}

\newtheorem{lemma}[theorem]{Lemma}
\newtheorem{claim}[theorem]{Claim}

\newtheorem{corollary}[theorem]{Corollary}

\theoremstyle{remark}

\theoremstyle{definition}
\newtheorem{definition}[theorem]{Definition}

\newcommand{\beq}{\begin{eqnarray}}
\newcommand{\eeq}{\end{eqnarray}}

\newcommand{\ket}[1]{|#1\rangle}
\newcommand{\bra}[1]{\langle#1|}
\newcommand{\proj}[1]{\ket{#1}\!\bra{#1}}
\newcommand{\Tr}{\mbox{\rm Tr}}
\newcommand{\Id}{\ensuremath{\mathop{\rm Id}\nolimits}}
\newcommand{\Es}[1]{\ensuremath{\mathop{\textsc{E}}_{#1}}}

\newcommand{\reg}[1]{{\textsf{#1}}}
\newcommand{\ol}[1]{\overline{#1}}

\newcommand{\F}{\ensuremath{\mathbb{F}}}

\newcommand{\R}{\ensuremath{\mathbb{R}}}
\newcommand{\Z}{\ensuremath{\mathbb{Z}}}

\newcommand{\mA}{\ensuremath{\mathcal{A}}}
\newcommand{\mB}{\ensuremath{\mathcal{B}}}

\newcommand{\mK}{\ensuremath{\mathcal{K}}}
\newcommand{\mS}{\ensuremath{\mathcal{S}}}

\newcommand{\mX}{\ensuremath{\mathcal{X}}}
\newcommand{\mY}{\ensuremath{\mathcal{Y}}}

\newcommand{\mH}{\mathcal{H}}

\DeclareMathOperator{\poly}{poly}

\newcommand{\eps}{\varepsilon}

\newcommand{\aux}{\mathsf{aux}}

\newcommand{\norm}[1]{\left\| {#1} \right\|}

\newcommand{\sync}{\textrm{sync}}

\newcommand{\game}{\mathfrak{G}}
\newcommand{\strategy}{\mathscr{S}}

\def\*#1{\mathbf{#1}}

\begin{document}

\title{Almost synchronous quantum correlations}
\author{Thomas Vidick\thanks{Email: \texttt{vidick@caltech.edu}}\\
Department of Computing and Mathematical Sciences, \\California Institute of Technology, \\USA. } 
\date{\today}
\maketitle


\begin{abstract}
The study of quantum correlation sets initiated by Tsirelson in the 1980s and originally motivated by questions in the foundations of quantum mechanics has more recently been tied to questions in quantum cryptography, complexity theory, operator space theory, group theory, and more. Synchronous correlation sets introduced in [Paulsen et.\ al, JFA 2016] are a subclass of correlations that has proven particularly useful to study and arises naturally in applications. We show that any correlation that is almost synchronous, in a natural $\ell_1$ sense, arises from a state and measurement operators that are well-approximated by a convex combination of projective measurements on a maximally entangled state. This extends a result of  [Paulsen et.\ al, JFA 2016] which applies to exactly synchronous correlations. Crucially, the quality of approximation is independent of the dimension of the Hilbert spaces or of the size of the correlation. Our result allows one to reduce the analysis of many classes of nonlocal games, including rigidity properties, to the case of strategies using maximally entangled states  which are generally easier to manipulate. 
\end{abstract}

\section{Introduction}



For finite sets $\mX,\mY,\mA$ and $\mB$ a \emph{quantum correlation} is an element of the set 
\begin{align*}
 C_q(\mX,\mY,\mA,\mB)\;=\; & \{ \big( \bra{\psi} A^x_a \otimes B^y_b \ket{\psi} \big)_{xyab} \,:\; \ket{\psi} \in \mH_A \otimes \mH_B\;, \\
&\qquad\qquad\forall x\in \mX,\,y\in \mY,\; \{A^x_a\}_{a\in \mA} \text{ POVM on }\mH_A, \; \{B^y_b\}_{b\in \mB}\text{ POVM on }\mH_B \big\}\;,
\end{align*}
where $\mH_A$ and $\mH_B$ range over all finite-dimensional Hilbert spaces and a POVM (positive operator-valued measure) on a Hilbert space $\mH$ is a collection of positive semidefinite operators on $\mH$ that sum to identity. For each $\mX,\mY,\mA,\mB$ the set $C_q(\mX,\mY,\mA,\mB)$ is convex, as can be seen by taking direct sums, but there are $\mX,\mY,\mA,\mB$ such that it is not closed~\cite{slofstra2019set}. We write $C_q$ for the union of $C_q(\mX,\mY,\mA,\mB)$ over all finite $\mX,\mY,\mA$ and $\mB$. 

 A \emph{strategy} is a tuple $\strategy=(\ket{\psi},A,B)$ such that $\ket{\psi} \in \mH_A\otimes \mH_B$ is a state (i.e.\ a unit vector), $A = \{A^x_a\}$  a collection of POVM on $\mH_A$ and $B=\{B^y_b\}$ a collection of POVM on $\mH_B$. (The finite-dimensional Hilbert spaces $\mH_A$ and $\mH_B$ as well as the index sets $\mX,\mY,\mA$ and $\mB$ are generally left implicit in the notation.) Given a strategy $\strategy$ we say that $\strategy$ \emph{induces} the correlation $(C_{x,y,a,b} = \bra{\psi}A^x_a \otimes B^y_b\ket{\psi})_{xyab}$. 

The study of the set of quantum correlations $C_q$ and its relation to the set of classical correlations $C$, defined as the convex hull of those correlations that can be induced using a state $\ket{\psi}$ that is a tensor product 
$\ket{\psi} = \ket{\psi_A}\otimes \ket{\psi_B} \in \mH_A\otimes \mH_B$,
is of importance in the foundations of quantum mechanics. The fact that $C\subsetneq C_q$, as first shown by Bell~\cite{bell1964einstein} and often termed ``quantum nonlocality,'' underlies the field of device-independent quantum cryptography and gives rise to the study of entanglement witnesses, protocols for delegated quantum computation, and questions in quantum complexity theory; we refer to~\cite{brunner2014bell} for references. Following the foundational work of Tsirelson~\cite{tsirelson1993some} multiple variants of the set of quantum correlations have been introduced and their study is connected to a range of problems in mathematics including operator space theory~\cite{fritz2012tsirelson,junge2011connes}, group theory~\cite{slofstra2018entanglement} and combinatorics~\cite{manvcinska2020quantum}. 

In this paper we consider a subset of $C_q$ introduced in~\cite{paulsen2016estimating} and called the \emph{synchronous set} $C_q^s$. It is defined as the union of all $C_q^s(\mX,\mA)$ where $C_q^s(\mX,\mA)$ is the subset of $C_q(\mX,\mX,\mA,\mA)$ that contains all those correlations $C$ which satisfy $C_{x,x,a,b} = 0$ whenever $a\neq b$. This set arises naturally in the study of certain classes of nonlocal games. In general a nonlocal game $\game$ is specified by a distribution $\nu$ on $\mX\times \mY$ and a function $D:\mX\times\mY\times\mA\times\mB\to \{0,1\}$. A nonlocal game gives rise to a linear function on $C_q(\mX,\mY,\mA,\mB)$ through the quantity 
\[ \omega_q(\game;C) \,=\, \sum_{x,y} \,\nu(x,y) \,\sum_{a,b}\, D(x,y,a,b) C_{x,y,a,b} \;.\]
Given a game $\game$ one is interested in its \emph{quantum value} $\omega_q(\game)$, which is defined as the supremum over all $C\in C_q$ of $\omega_q(\game;C)$. A game $\game$ such that $\mX=\mY$, $\mA=\mB$, $\nu(x,x)>0$ for all $x$ and $D(a,b|x,x)=0$ for all $x$ and $a\neq b$ is called a \emph{synchronous game}. Any such game has the property that $\omega_q(\game;C)=1$ can only be obtained by a $C \in C_q^s$. Synchronous games arise  naturally in applications; see e.g. the classes of graph homomorphism games~\cite{mancinska2014graph} or linear system games~\cite{cleve2014characterization}. (Linear system games are projection games, which can be turned into synchronous games by taking their ``square''; see~\cite{kim2018synchronous}.)
The set $C_q^s$ retains most of the interesting geometric aspects of $C_q$, and in particular it is convex and non-closed~\cite{kim2018synchronous}. 

A key property of synchronous correlations that makes them more amenable to study is the following fact shown in~\cite{paulsen2016estimating}. For every synchronous correlation $C$ there is a family of strategies $\{\strategy^{\lambda} = (\ket{\psi_{\lambda}},A^{\lambda},B^{\lambda})\}_{\lambda\in \Lambda}$ and a measure $\mu$ on $\Lambda$ such that for each $\lambda$, $\ket{\psi_\lambda}=\ket{\psi_{me}(d_\lambda)}$ with 
\begin{equation}\label{eq:def-me}
\ket{\psi_{me}(d_\lambda)} \,=\, \frac{1}{\sqrt{d_\lambda}}\sum_{i=1}^{d_\lambda} \,\ket{u_i}\ket{u_i}\,\in\,\mH_\lambda\otimes \mH_\lambda\;,
\end{equation}
where $d_\lambda =\dim(\mH_\lambda)$ and $\{\ket{u_i}: \,1\leq i \leq d_\lambda\}$ is an orthonormal family in $\mH_\lambda$, each measurement $\{A^{\lambda,x}_a\}$ and $\{B^{\lambda,y}_b\}$ consists entirely of projections, and moreover for all $x,y,a,b$ we have
\begin{equation}\label{eq:ce-1}
 C_{x,y,a,b} \,=\, \int_\lambda  \bra{\psi_\lambda}A^{\lambda,x}_a \otimes B^{\lambda,y}_b\ket{\psi_\lambda}   \, d\mu(\lambda)\;.
\end{equation}
When $\ket{\psi_{me}(d_\lambda)} $ takes the form in~\eqref{eq:def-me}
 we can express
\begin{equation}\label{eq:ce-2}
 \bra{\psi_\lambda}A^{\lambda,x}_a \otimes B^{\lambda,y}_b\ket{\psi_\lambda} \,=\, \frac{1}{d_\lambda} \,\Tr\big({A}^{\lambda,x}_a ({B}^{\lambda,y}_b)^T\big)\;,
\end{equation}
where $\Tr(\cdot)$ is the usual matrix trace and $X^T$ denotes the transpose with respect to the basis $\{\ket{u_i}\}$. The fact that synchronous correlations are ``tracial'' in the sense given by~\eqref{eq:ce-1} and~\eqref{eq:ce-2} contributes largely to their appeal.
 In contrast there are correlations $C\in C_q$ such that $C$ cannot be induced, even approximately, using a convex combination of strategies using states of the form~\eqref{eq:def-me} in any dimension; see~\cite{vidick2011more} for an example.  Such correlations tend to be more difficult to study and their main interest lies in their existence, e.g.\ they can provide entanglement witnesses for states that are not maximally entangled.

\paragraph{Our results.}
We consider strategies $\strategy=(\ket{\psi},A,B)$ that are \emph{almost} synchronous, where the default to synchronicity is measured by the quantity
\begin{equation}\label{eq:delta-sync}
 \delta_\sync(C;\nu) \,=\, \Es{x\sim\nu} \,\sum_{a\neq b}\, C_{x,x,a,b} \;,
\end{equation}
where $\nu$ is some distribution on $\mX$. This averaged $\ell_1$ distance is motivated by applications to nonlocal games, which we describe below. 
 Informally our main result is that any strategy $\strategy$ that induces a correlation $C$ is well-approximated by a convex combination of strategies $\strategy^\lambda$ each using a maximally entangled state, where the approximation is controlled by $\delta_\sync(C;\nu)$ for any $\nu$ ($\nu$ also enters in the measure of approximation between $\strategy$ and the $\strategy^\lambda$) and, crucially for applications, does not depend on the dimension of $\ket{\psi}$ or the size of the sets $\mX$ and $\mA$. In particular each $\strategy^\lambda$ gives rise to a synchronous correlation $C^\lambda$ such that $\int_\lambda C^\lambda \approx C$ in a suitable $\ell_1$ sense. Moreover, and crucially for the applications that we describe next, specific structural properties of the $\strategy^\lambda$, such as algebraic relations between some of the measurement operators, can be transferred to the strategy $\strategy$. A simplified version of our theorem specialized to the case of a single measurement can be stated as follows. 

\begin{theorem*}\label{thm:main-intro}
There are universal constants $c,C>0$ such that the following holds. Let $\mH$ be a finite-dimensional Hilbert space and $\ket{\psi}\in \mH\otimes \mH$ a state. Then there is a finite set $\Lambda$, a distribution $\mu$ on $\Lambda$, and for each $\lambda\in\Lambda$ a state $\ket{\psi_\lambda}$ that is maximally entangled on a subspace $\mH_\lambda \otimes \mH_\lambda \subseteq \mH \otimes \mH$ such that letting $\rho$ be the reduced density of $\ket{\psi}$ on the first factor and $\rho_\lambda$ the reduced density of $\ket{\psi_\lambda}$ on $\mH_\lambda\subseteq\mH$, 
\begin{equation}\label{eq:rho-plambda-i}
 \rho \,=\, \Es{\lambda \sim \mu} \big[\,\rho_\lambda\,\big]\;.
\end{equation}
Moreover, let $\mA$ be a finite set and $\{A_a\}_{a\in \mA}$ an arbitrary measurement on $\mH$. Then there is
a projective measurement $\{A^{\lambda}_a\}$ on $\mH_\lambda$ such that
\begin{equation}\label{eq:lambda-strat-i}
\Es{\lambda\sim\mu} \Big[  \sum_a  \big\| \big(A_a - A^{\lambda}_a\big)\otimes \Id \ket{\psi_\lambda} \big\|^2\,\Big]\,\leq\,C \,\Big(1-\sum_a\bra{\psi} A_a \otimes A_a \ket{\psi}\Big)^c\;.
\end{equation}
\end{theorem*}

For the complete statement and additional remarks, see Theorem~\ref{thm:main}. The first part of the theorem,~\eqref{eq:rho-plambda-i}, is very simple to obtain; it is the second part that is meaningful. In particular, since $\ket{\psi_\lambda}$ is a maximally entangled state the approximation on the left-hand side can be seen as a form of weighted approximation over certain (overlapping) diagonal blocks of $A$.
The fact that the spaces $\mH_\lambda$ and the states $\ket{\psi_\lambda}$ depend on $\ket{\psi}$ only allows us to apply the theorem repeatedly for different measurements in order to decompose an arbitrary strategy as a convex combination of projective maximally entangled strategies, with the right-hand side in~\eqref{eq:lambda-strat-i} replaced by $C\delta_\sync(C;\nu)^c$ for a $\nu$ of one's choice (which naturally will also appear on the left-hand side). 

A consequence of the theorem is that any  $C \in \overline{C_q}$ which is also synchronous can be approximated by elements of $C_q^s$; this is because any sequence of approximations to $C$ taken from $C_q$ must, by definition, be almost synchronous and so Theorem~\ref{thm:main} can be applied. (For this observation it is crucial that the approximation provided in Theorem~\ref{thm:main} does not depend on the dimension of the Hilbert spaces; however, it could depend on the size of $C$.) This particular application was already shown in~\cite[Theorem 3.6]{kim2018synchronous}. 

Our result and its formulation are motivated by the study of nonlocal games. 
For a strategy $\strategy$ we write $\omega_q(\game;\strategy)$ for $\omega_q(\game;C)$ where $C$ is the correlation induced by $\strategy$. Recall that the game value $\omega_q(\game)$ is the supremum over all strategies of $\omega_q(\game;\strategy)$. The fact that the supremum is taken over $C_q$ and not $C_q^s$ is motivated by applications to entanglement tests, cryptography, and complexity theory, as in those contexts there is no a priori reason to enforce hard constraints of the form $C_{x,x,a,b}=0$; indeed, such a constraint cannot be verified with absolute confidence in any statistical test. 

Given a game and a strategy $\strategy$ for it it is possible to obtain statistical confidence that $\omega_q(\game;\strategy)\geq \omega_q(\game)-\eps$ for finite $\eps>0$ by playing the game many times. For this reason the characterization of nearly-optimal strategies plays a central role in applications of nonlocality. Recall that a synchronous game has the property that $D(x,x,a,b) = 0$ whenever $a\neq b$. Given a synchronous game $\game$ such that furthermore $\omega_q(\game)=1$ it follows that any strategy $\strategy$ for $\game$ such that $\omega_q(\game;\strategy)\geq \omega_q(\game)-\eps$ must satisfy $\delta_\sync(\strategy;\nu_{diag}) = O(\eps)$, where $\nu_{diag}(x)=\nu(x,x)/(\sum_{x'} \nu(x',x'))$ and the constant implicit in the $O(\cdot)$ notation will in general depend on the weight that $\nu$ places on the diagonal. (In particular a better bound on $\delta_{sync}$ will be obtained in cases when the distribution $\nu$ is \emph{not} a uniform distribution, as the uniform distribution places weight $\approx \frac{1}{|\mX|}$ on the diagonal, which can be quite small.) Thus nearly-optimal strategies in synchronous games give rise to almost synchronous correlations. This conclusion may also hold for games that are not necessarily synchronous, for example because the sets $\mX$ and $\mY$ are disjoint; an example is the class of projection games that we consider in Section~\ref{sec:projection}. Examples of projection games include linear system games~\cite{cleve2014characterization} and games such as the low-degree test~\cite{ml2020} that play an important role in complexity theory. 

Given the importance of studying nearly-optimal strategies, the fact that for many games any nearly-optimal strategy is almost synchronous ought to be useful. Our work allows one to reduce the analysis of almost synchronous strategies to that of exactly synchronous strategies in a broad variety of settings. The most direct application of our results is to the study of the  phenomenon of rigidity, which seeks to extract necessary conditions of any strategy that is nearly-optimal for a certain game. Informally, our results imply that a general rigidity result for a synchronous game can be obtained in an automatic manner from a rigidity result that applies only to perfectly synchronous strategies. In order for the implication to not lose factors depending on the size of the game in the approximation quality for the rigidity statement it is sufficient that a high success probability in the game implies a low $\delta_\sync(\strategy;\nu)$ for $\nu$ the marginal distribution on either player's questions in the game; see Corollary~\ref{cor:main} and the remarks following it for further discussion. 
 To give just one example, the entire analysis carried out in the recent~\cite{ji2020mip} could be simplified by making all calculations with the maximally entangled state only, making manipulations of the ``state-dependent distance'' far easier to carry out. We refer Section~\ref{sec:rigid} for a precise formulation of how our main result can be used in this context as well as another application, to showing algebraic relations between measurement operators. 

\paragraph{Discussion.} Given an almost synchronous strategy $\strategy=(\ket{\psi},A,B)$ it is not hard to show that the state and operators that underlie the strategy behave in an ``approximately'' cyclic manner, e.g.\ letting $\rho_A$ denote the reduced density of $\ket{\psi}$ on $\mH_A$ it holds that $\|A^x_a \rho_A - \rho_A A^x_a \|_1 \approx 0$ for all $x,a$ where $\|\cdot\|_1$ denotes the Schatten-$1$ norm; see e.g.~\cite[Lemma 3.7]{prakash} for a precise statement. The strength of our result lies in showing that such relations imply an approximate decomposition in terms of maximally entangled strategies, where crucially the approximation quality does not depend on the dimension of the Hilbert space nor on the size of the sets $\mX,\mY,\mA$ or $\mB$. A similar decomposition implicitly appears in~\cite{slofstra2018entanglement}, where it is used to reduce the analysis of nearly-optimal strategies for a specific linear system game to the case of maximally entangled strategies; in the context of that paper the reduction is motivated by a connection with the study of approximate representations of a certain finitely presented group. The main technical ingredient that enables the reduction in~\cite{slofstra2018entanglement} is also the main ingredient in the present paper, which can be seen as a direct generalization of the work done there. Informally the key idea is to write any density matrix $\rho$ as a convex combination of projections $\chi_{\geq \sqrt{\lambda}}(\rho)$, where $\lambda$ is any non-negative real and $\chi_{\geq \sqrt{\lambda}}$ is the indicator of the interval $[\sqrt{\lambda},+\infty)$; see Lemma~\ref{lem:int}. The main additional observation needed is a calculation which originally appears in~\cite{connes1976classification} and is restated as Lemma~\ref{lem:connestrick} below; informally, the calculation allows to transfer approximate commutation conditions such as those obtained in~\cite[Lemma 3.7]{prakash} for any almost synchronous strategy to the same conditions, evaluated on the matrix $\chi_{\geq \sqrt{\lambda}}(\rho)$. The latter is a scaled multiple of the identity and is thus directly related to a maximally entangled state. 


\section{Preliminaries}

\subsection{Notation}

We use $\mX,\mY,\mA,\mB$ to denote finite sets. 
We use $\mH$ to denote a finite-dimensional Hilbert space, which we generally endow with a canonical orthonormal basis $\{\ket{i}|\, i\in\{1,\ldots,d\}\}$ with $d=\dim(\mH)$. We use $\|\cdot\|$ to denote the operator norm (largest singular value) on $\mH$. $\Tr(\cdot)$ is the trace on $\mH$ and $\|\cdot\|_F$ the Frobenius norm $\|X\|_F = \Tr(X^\dagger X)^{1/2}$ for any operator $X$ on $\mH$, where $X^\dagger$ is the conjugate-transpose. 
A positive operator-valued measure (POVM), or \emph{measurement} for short, on $\mH$ is a  finite collection of positive semidefinite operators $\{A_a\}_{a\in \mA}$ such that $\sum_a A_a=\Id$. A measurement $\{A_a\}$ is \emph{projective} if each $A_a$ is a projection. 

We use $\poly(\delta)$ to denote any real-valued function $f$ such that there exists constants $C,c>0$ with $|f(\delta)|\leq C\delta^c$ for all non-negative real $\delta$. The precise function $f$ as well as the constants $c,C$ may differ each time the notation is used. For a distribution $\nu$ on a finite set $\mX$ we write $\Es{x\sim \nu}$ for the expectation with respect to $x$ with distribution $\nu$. 

\subsection{Strategies, correlations and games}

\begin{definition}[Strategies and correlations]
A \emph{strategy} $\strategy$ is a tuple $(\ket{\psi},A,B)$ where $\ket{\psi}\in\mH_\reg{A}\otimes \mH_\reg{B}$ is a quantum state and $A = \{A^x_a\}$ (resp. $B=\{B^y_b\}$) is a collection of measurements on $\mH$ indexed by $x\in \mX$ and with outcomes $a\in \mA$ (resp. $y\in\mY$ and $b\in\mB$). Any strategy induces a \emph{correlation}, which is the collection of real numbers 
\[ C_{xyab} \,=\, \bra{\psi} A^x_a \otimes B^y_b \ket{\psi}\;,\qquad \forall (x,y) \in \mX\times \mY\;,\quad \forall (a,b)\in \mA\times \mB\;.\]
The set of all correlations that arise from strategies of this form is denoted 
\[ C_q(\mX,\mY,\mA,\mB) \subseteq \R^{|\mX||\mY||\mA||\mB|}\;.\]
\end{definition}

\begin{definition}[Synchronous correlations]
For finite sets $\mX$ and $\mA$ a correlation $C= (C_{x,y,a,b}) \in C_q(\mX,\mX,\mA,\mA)$ is called \emph{synchronous} if $C_{x,x,a,b} = 0$ for all $x\in \mX$ and $a,b\in \mA$ such that $a\neq b$. Given a distribution $\nu$ on $\mX$ recall the definition of $\delta_{\sync}(C;\nu)$ in~\eqref{eq:delta-sync}. 
Given a strategy $\strategy$ we also write $\delta_{\sync}(\strategy;\nu)$ for $\delta_{\sync}(C;\nu) $ where $C$ is the correlation induced by $\strategy$.  
\end{definition}

\begin{definition}[PME strategies]
 A strategy $\mS = (\ket{\psi},A,B)$ is \emph{symmetric} if $\mX=\mY$, $\mA=\mB$,   $\ket{\psi}\in \mH\otimes \mH$ takes the form
\begin{equation}\label{eq:psi-sym}
 \ket{\psi} \,=\, \sum_i \sqrt{\lambda_i} \ket{u_i} \ket{u_i} \,\in\,\mH\otimes \mH\;,
\end{equation}
where the $\lambda_i$ are non-negative and $\{\ket{u_i}\}$ orthonormal, and for every $x,a$, $A^x_a=(B^x_a)^T$ with the transpose being taken with respect to the $\{\ket{u_i}\}$. Note that this implies that $\mH_\reg{A}=\mH_\reg{B}$ and that $\ket{\psi}$ has the same reduced density on either subsystem. For a symmetric strategy we write it as $\strategy = (\ket{\psi},A)$. A strategy is \emph{projective} if all $A^x_a$ and $B^y_b$ are projections. It is \emph{maximally entangled} if $\ket{\psi}$ is a maximally entangled state~\eqref{eq:def-me} on $\mH_\reg{A}\otimes \mH_\reg{B}$. We use the acronym ``PME'' to denote ``symmetric projective maximally entangled''.
\end{definition}

Observe that a correlation defined from a PME strategy is synchronous. (A converse to this statement is shown in~\cite{paulsen2016estimating}, i.e.\ every synchronous strategy arises from (a convex combination of) PME strategies.) To see this, first recall \emph{Ando's formula}: for any $X,Y$ and $\ket{\psi}\in \mH\otimes \mH$ of the form~\eqref{eq:psi-sym} with reduced density $\rho$ it holds that
\begin{equation}\label{eq:ando}
 \bra{\psi} X \otimes Y \ket{\psi}\,=\, \Tr\big(X\rho^{1/2} Y^T \rho^{1/2}\big)\;,
\end{equation}
where the transpose is taken with respect to the basis $\{\ket{u_i}\}$ as in~\eqref{eq:psi-sym}. Now for a PME strategy $\strategy = (\ket{\psi},A)$ and any $\nu$ write
\begin{align*}
\delta_{\sync}(\strategy,\nu) &= 1- \Es{x\sim \nu} \sum_a \, \bra{\psi} A^x_a \otimes (A^x_a)^T \ket{\psi} \\
&= 1-  \Es{x\sim \nu} \sum_a \, \bra{\psi} (A^x_a)^2 \otimes \Id \ket{\psi}\\
&= 1-  \Es{x\sim \nu} \sum_a \,  \bra{\psi} A^x_a \otimes\Id \ket{\psi}\\
&= 0\;,
\end{align*}
where the first equality is by definition, the second uses~\eqref{eq:ando} together with the fact that for a PME strategy the reduced density of $\ket{\psi}$ on either system is proportional to the identity, hence commutes with any operator, the third equality uses that all $A^x_a$ are projections and the last that they sum to identity. 

We reproduce a definition from~\cite{prakash}.

\begin{definition}[Local $(\eps,\nu)$-dilation]\label{def:dilation} Given $\eps \geq 0$, a distribution $\nu$ on $\mX\times \mY$ and two strategies $\strategy = (\ket{\psi},A,B)$ and $\tilde{\strategy} = (\ket{\tilde{\psi}},\tilde{A},\tilde{B})$ we say that $\tilde{\strategy}$ is a \emph{local $(\eps,\nu)$-dilation} of $\strategy$ if there exists isometries $V_A: \mH_A \to \tilde{\mH}_A \otimes \mK_A$ and $V_B: \mH_B \to \tilde{\mH}_B \otimes \mK_B$ and a state $\ket{\aux} \in \mK_A \otimes \mK_B$ such that 
\begin{align*}
 \big\|(V_A \otimes V_B) \ket{\psi}  - \ket{\tilde{\psi}} \otimes \ket{\aux} \big\| \,&\leq\, \eps\;,\\
\Big(\Es{(x,y)\sim \nu} \sum_{a,b} \big\| (V_A\otimes V_B) \big( A^x_a \otimes B^y_b \big) \ket{\psi} -\big( \big(\tilde{A}^x_a \otimes \tilde{B}^y_b\big) \ket{\tilde{\psi}} \big)\otimes \ket{\aux}\big\|^2 \Big)^{1/2}\,&\leq\, \eps\;.
\end{align*}
\end{definition}

In~\cite{prakash} the second condition is required for all $x,y,a,b$. We require it to hold in an averaged sense only because this is more natural when seeking approximations that are independent of the size of the sets $\mX,\mY,\mA,\mB$, as is the case in the context of this paper.

The following lemma implies that for any correlation $C\in C_q$ there is a projective (but not necessarily maximally entangled) strategy that realizes it.

\begin{lemma}[Naimark dilation]\label{lem:naimark}
Let $\ket{\psi}$ be a state in $\mH_{\mathrm{A}} \otimes \mH_{\mathrm{B}}$.
Let $A = \{A^x_a\}$ be a measurement on $\mH_{\mathrm{A}}$
and $B=\{B^y_b\}$ a measurement on $\mH_{\mathrm{B}}$.
Then there exists
Hilbert spaces $\mH_{\mathrm{A}_{\mathsf{aux}}}$ and $\mH_{\mathrm{B}_{\mathsf{aux}}}$,
 a state $\ket{\mathsf{aux}} \in \mH_{\mathrm{A}_{\mathsf{aux}}} \otimes \mH_{\mathrm{B}_{\mathsf{aux}}}$,
 and two projective measurements  $\widehat{A} = \{\widehat{A}^x_a\}$ and $\widehat{B}=\{\widehat{B}^y_b\}$
acting on $\mH_{\mathrm{A}} \otimes \mH_{\mathrm{A}_{\mathsf{aux}}}$
and $\mH_{\mathrm{B}} \otimes \mH_{\mathrm{B}_{\mathsf{aux}}}$,
respectively,
such that the following is true.
If we let $\ket{\widehat{\psi}} = \ket{\psi} \otimes \ket{\mathsf{aux}}$
then for all $x,y, a, b$,
 \begin{equation*}
 \bra{\psi} A^x_a \otimes B^y_b \ket{\psi}
 = \bra{\widehat{\psi}} \widehat{A}^x_a \otimes \widehat{B}^y_b \ket{\widehat{\psi}}\;.
 \end{equation*}
 In addition $\ket{\mathsf{aux}}$ is a \emph{product state},
 meaning that we can write it as
$ \ket{\mathsf{aux}}= \ket{\mathsf{aux}_{\mathrm{A}}} \otimes \ket{\mathsf{aux}_{\mathrm{B}}}$,
for $\ket{\mathsf{aux}_{\mathrm{A}}}$ in $\mH_{\mathrm{A}_{\mathsf{aux}}}$
and $\ket{\mathsf{aux}_{\mathrm{B}}}$ in $\mH_{\mathrm{B}_{\mathsf{aux}}}$.
\end{lemma}

\begin{definition}
A \emph{nonlocal game} (or \emph{game} for short) $\game$ is specified by a tuple $(\mX,\mY,\mA,\mB,\nu,D)$ of finite \emph{question sets} $\mX$ and $\mY$, finite \emph{answer sets} $\mA$ and $\mB$, a distribution $\nu$ on $\mX\times\mY$ and a \emph{decision predicate} $D:\mX\times \mY\times \mA\times \mB\to \{0,1\}$ that we conventionally write as $D(a,b|x,y)$ for $(x,y)\in \mX\times \mY$ and $(a,b)\in \mA\times \mB$. The game is \emph{symmetric} if $\mX=\mY$, $\mA=\mB$, $\nu(x,y)=\nu(y,x)$ for all $x,y\in \mX\times\mY$, and for all $a,b,x,y$, $D(a,b|x,y)=D(b,a|y,x)$. In this case we write $\game = (\mX,\mA,\nu,D)$. We often abuse notation and also use $\nu$ to denote the marginal distribution of $\nu$ on $\mX$. 
\end{definition}

\begin{definition}
Given a game $\game=(\mX,\mY,\mA,\mB,\nu,D)$ and a strategy $\strategy = (\ket{\psi},A)$ in $\game$, the \emph{success probability} of $\strategy$ in $\game$ is 
\[\omega(\game;\strategy)\,=\, \Es{(x,y)\sim \nu} \sum_{a,b} D(a,b|x,y) \bra{\psi} A^x_a \otimes B^y_b \ket{\psi}\;.\]
\end{definition}

\subsection{Consistency}

We show elementary and generally well-known lemmata that will be useful in the proofs. The first lemma relates two different measures of state-dependent distance between measurements on $\mH$. 

\begin{lemma}\label{lem:am-cons}
Let $\gamma,\delta \geq 0$. Let $\mH$ be a Hilbert space and $\ket{\psi}\in\mH\otimes \mH$ a state. Let $\mX$ be a finite set and for each $x\in\mX$, $\{A^x_a\}_{a\in\mA}$ a  projective measurement and $\{M^x_a\}_{a\in\mA}$ an arbitrary measurement on $\mH$. Let $\mu$ be a distribution on $\mX$. Let
\begin{equation}\label{eq:am-cons-0a}
\delta \,=\,1-\Es{x\sim \mu} \sum_{a\in\mA} \bra{\psi} A^x_a \otimes (A^x_a)^T \ket{\psi}\quad\text{and}\quad \gamma = 1 -  \Es{x\sim \mu}\sum_{a\in\mA} \bra{\psi} A^x_a \otimes (M^x_a)^T \ket{\psi}\;.
\end{equation}
Then 
\begin{equation}\label{eq:am-cons-0b}
(\gamma - \delta)^2 \,\leq\, \Es{x\sim \mu} \sum_{a\in\mA} \bra{\psi}\Id \otimes \big( A^x_a - M^x_a \big)^2 \ket{\psi} \,\leq\, 2\,\gamma + 2\sqrt{2\delta}\;.
\end{equation}
\end{lemma}

\begin{proof}
We start with the left inequality:
\begin{align*}
\gamma &= 1 -  \Es{x\sim \mu}\sum_a \bra{\psi} A^x_a \otimes (M^x_a)^T \ket{\psi}\\
&= \Es{x\sim \mu}\sum_a \bra{\psi} A^x_a \otimes \big(A^x_a-M^x_a\big)^T \ket{\psi} + 1 - \Es{x\sim \mu}\sum_a \bra{\psi} A^x_a \otimes (A^x_a)^T \ket{\psi}\\
&\leq \Big(\Es{x\sim \mu}\sum_a \bra{\psi} \Id \otimes \big(A^x_a-M^x_a\big)^2 \ket{\psi}\Big)^{1/2}\Big(\Es{x\sim \mu}\sum_a \bra{\psi} (A^x_a)^2 \otimes \Id \ket{\psi}\Big)^{1/2} + \delta\\
&\leq \Big(\Es{x\sim \mu}\sum_a \bra{\psi} \Id \otimes \big(A^x_a-M^x_a\big)^2 \ket{\psi}\Big)^{1/2} + \delta \;,
\end{align*}
where the first inequality is Cauchy-Schwarz and the last uses that for each $x\in\mX$, $\{A^x_a\}_a$ is a measurement. 

For the right inequality, 
\begin{align*}
\Es{x\sim \mu}\sum_a &\bra{\psi} \Id \otimes \big(A^x_a-M^x_a\big)^2 \ket{\psi}\\
&= \Es{x\sim \mu}\sum_a \bra{\psi} \Id \otimes\big( (A^x_a)^2 + (M^x_a)^2 \big)\ket{\psi} - 2\Re\Big(  \Es{x\sim \mu}\sum_a \bra{\psi} \Id \otimes A^x_a M^x_a \ket{\psi}\Big)\\
&\leq 2 - 2\Re\Big(  (1-\gamma)-  \Es{x\sim \mu}\sum_a \bra{\psi} \big(A^x_a \otimes \Id - \Id \otimes (A^x_a)^T\big)(\Id\otimes  (M^x_a)^T) \ket{\psi}\Big)\\
&\leq 2\,\gamma + 2\Big(\Es{x\sim \mu}\sum_a \bra{\psi} \big(A^x_a \otimes \Id - \Id \otimes (A^x_a)^T\big)^2 \ket{\psi}\Big)^{1/2}\Big(\Es{x\sim \mu}\sum_a \bra{\psi} \Id\otimes ((M^x_a)^T)^2 \ket{\psi}\Big)^{1/2}\\
&\leq 2\,\gamma + 2\sqrt{2\delta}\;,
\end{align*}
where the first inequality uses that $\{A^x_a\}_a$ and $\{M^x_a\}_a$ are measurements, the second uses the Cauchy-Schwarz inequality and the last the definition of $\delta$ and that for each $x\in\mX$, $\{A^x_a\}$ is a projective measurement and $\{M^x_a\}_a$ is a measurement. 
\end{proof}

Given a density matrix $\rho$ on $\mH$, define the \emph{canonical purification} of $\rho$ as the state 
\[ \ket{\psi} \,=\, \sum_i \sqrt{\lambda_i} \ket{u_i} \ket{u_i}\;,\]
where $\rho = \sum_i \lambda_i \proj{u_i}$ is the spectral decomposition. 

\begin{lemma}\label{lem:psi-cs}
Let $\ket{\psi}\in \mH_A \otimes \mH_B$ and $\{A_a\}$ and $\{B_a\}$ measurements on $\mH_A$ and $\mH_B$ respectively. Let $\rho_A$ and $\rho_B$ be the reduced density of $\ket{\psi}$ on $\mH_A$ and $\mH_B$ respectively. Let $\ket{\psi_A}\in\mH_A\otimes\mH_A$ and $\ket{\psi_B}\in\mH_B\otimes\mH_B$ be the canonical purifications of $\rho_A$ and $\rho_B$ respectively. Then 
\[ \sum_a \bra{\psi} A_a \otimes B_a \ket{\psi} \,\leq\,\Big( \sum_a \bra{\psi_A} A_a \otimes A_a^T \ket{\psi_A} \Big)^{1/2}\Big( \sum_a \bra{\psi_B} B_a \otimes B_a^T \ket{\psi_B} \Big)^{1/2}\;.\]
\end{lemma}

\begin{proof}
Let $\ket{\psi} = \sum_j \lambda_j \ket{u_j} \ket{v_j}$ be the Schmidt decomposition. Let $K = \sum_j \lambda_j \ket{u_j}\bra{v_j}$. Then 
\begin{align*}
 \sum_a \bra{\psi} A_a \otimes B_a \ket{\psi} 
&= \sum_a \Tr\big( A_a K \ol{B}_a K^\dagger \big)\\
&\leq \Big(\sum_a \Tr\big( A_a \sqrt{KK^\dagger} A_a \sqrt{KK^\dagger}\big) \Big)^{1/2} \Big(\sum_a \Tr\big( \ol{B_a} \sqrt{K^\dagger K } \ol{B_a} \sqrt{K^\dagger K }\big) \Big)^{1/2}\;,
\end{align*}
where the inequality is Cauchy-Schwarz. Using that $\rho_A = KK^\dagger$ and $\rho_B = K^\dagger K$, this concludes the proof. 
\end{proof}

The next lemma  gives conditions under which two strategies induce nearby correlations. 

\begin{lemma}\label{lem:close-cor}
Let $\strategy = (\ket{\psi},A,B)$ be  a strategy, $\hat{A}=\{\hat{A}^x_a\}$ a family of POVM on $\mH_A$, and let $\hat{\strategy} = (\ket{\psi},\hat{A},B)$. Let $\rho_A$ be the reduced density of $\ket{\psi}$ on $\mH_A$ and $\ket{\psi_A} \in \mH_A \otimes \mH_A$ the canonical purification of it. Let $\strategy_A = (\ket{\psi_A},A)$. 
Let $\nu$ be a distribution on $\mX\times\mY$ and $\delta = \delta_\sync(\strategy_A;\nu_A)$ where $\nu_A$ is the marginal of $\nu$ on $\mX$. Let 
\[\gamma \,=\, \Es{x\sim\nu_A} \,\sum_a \,\Tr\big( \big(A^x_a - \hat{A}^x_a \big)^2 \rho_A \big)\;.\]
Let $C$ be the correlation induced by $\strategy$ and $\hat{C}$ by $\hat{\strategy}$. Then 
\[ \Es{x,y\sim \nu} \sum_{a,b} \big| C_{x,y,a,b} - \hat{C}_{x,y,a,b}\big| \,\leq\, O\big(\delta + \sqrt{\gamma}\big) \;.\]
\end{lemma}

\begin{proof}
Conjugating the $B^y_b$ by a unitary if necessary we assume without loss of generality that the reduced densities of $\ket{\psi}$ on either subsystem satisfy $\rho_A = \rho_B$. Then $\ket{\psi_A} = \ket{\psi}$. 
As a first step in the proof we show that 
\begin{equation}\label{eq:cc-1}
\Es{x,y\sim \nu} \sum_{a,b} \big| C_{x,y,a,b} - \bra{\psi} (A^x_a)^2 \otimes  B^y_b  \ket{\psi}\big| \,\leq\, \delta\;.
\end{equation}
To show this, write
\begin{align}
\Es{x,y\sim \nu} \sum_{a,b} \big| C_{x,y,a,b} - \bra{\psi} (A^x_a)^2 \otimes  B^y_b  \ket{\psi}\big| &= 
\Es{x,y\sim \nu} \sum_{a,b} \bra{\psi} \big(A^x_a -  (A^x_a)^2 \big)\otimes  B^y_b  \ket{\psi} \notag\\
&\leq \Es{x\sim\nu_A} \sum_a  \bra{\psi} \big(A^x_a -  (A^x_a)^2 \big) \otimes \Id \ket{\psi}\notag\\
&= 1 - \Es{x\sim\nu_A} \sum_a  \bra{\psi} (A^x_a)^2 \otimes \Id \ket{\psi}\;,\label{eq:cc-2}
\end{align}
where the first step uses that $A^x_a -  (A^x_a)^2 \geq 0$ for all $x,a$, the second that $\sum_b B^y_b = \Id$ for all $y$ and the third that $\sum_a A^x_a = \Id$ for all $x$. Next we observe that 
\begin{align*}
1-\delta &= \Es{x\sim\nu_A} \sum_a  \bra{\psi} A^x_a\otimes (A^x_a)^T \ket{\psi}\\
&\leq \Big(\Es{x\sim\nu_A} \sum_a  \bra{\psi} (A^x_a)^2\otimes \Id \ket{\psi}\Big)^{1/2}\Big(\Es{x\sim\nu_A} \sum_a  \bra{\psi} \Id \otimes ((A^x_a)^T)^2\ket{\psi}\Big)^{1/2}\\
&= \Es{x\sim\nu_A} \sum_a  \Tr\big( (A^x_a)^2\rho_A \big)\;,
\end{align*}
where the inequality on the second line is Cauchy-Schwarz and the last line uses our assumption that $\rho_A = \rho_B$. Plugging back into~\eqref{eq:cc-2}, this shows~\eqref{eq:cc-1}. For the second step we show
\begin{equation}\label{eq:cc-3}
\Es{x,y\sim \nu} \sum_{a,b} \big| \bra{\psi} (A^x_a)^2 \otimes  B^y_b  \ket{\psi} - \bra{\psi} (\hat{A}^x_a)^2 \otimes  B^y_b  \ket{\psi}\big| \,\leq\, 2\sqrt{\gamma}\;.
\end{equation}
To show~\eqref{eq:cc-3} we first bound
\begin{align}
\Es{x,y\sim \nu} \sum_{a,b} \big| \bra{\psi} \big((A^x_a)^2 &- A^x_a \hat{A}^x_a \big)\otimes  B^y_b  \ket{\psi} \big|\notag\\
&\leq \Big(\Es{x,y\sim \nu} \sum_{a,b} \big| \bra{\psi} (A^x_a)^2 \otimes  B^y_b  \ket{\psi}  \Big)^{1/2}\Big(\Es{x,y\sim \nu} \sum_{a,b} \big| \bra{\psi} (A^x_a-\hat{A}^x_a)^2 \otimes  B^y_b  \ket{\psi}  \Big)^{1/2} \notag\\
&\leq \sqrt{\gamma}\;,\label{eq:cc-4a}
\end{align}
where the first inequality is Cauchy-Schwarz and the second bounds the first term by $1$ and the second by $\gamma$ using $\sum_b B^y_b = \sum_a A^x_a = 1$ and the definition of $\gamma$. An analogous calculation gives
\begin{equation}\label{eq:cc-4b}
\Es{x,y\sim \nu} \sum_{a,b} \big| \bra{\psi} \big(A^x_a \hat{A}^x_a - (\hat{A}^x_a)^2  \big)\otimes  B^y_b  \ket{\psi} \big| \,\leq\, \sqrt{\gamma}\;.
\end{equation}
Together,~\eqref{eq:cc-4a} and~\eqref{eq:cc-4b} give~\eqref{eq:cc-3}. Finally, the third step of the proof is given by the bound 
\begin{equation}\label{eq:cc-5}
\Es{x,y\sim \nu} \sum_{a,b} \big| \hat{C}_{x,y,a,b} - \bra{\psi} (\hat{A}^x_a)^2 \otimes  B^y_b  \ket{\psi}\big| \,\leq\, 2\big(\sqrt{\gamma} + \delta\big)\;.
\end{equation}
This is analogous to~\eqref{eq:cc-1} except that we rely on an estimate for consistency $\hat{\delta}$ of the $\{\hat{A}^x_a\}$. This can be obtained directly by using the left inequality of~\eqref{eq:am-cons-0b} in Lemma~\ref{lem:am-cons}, which letting $\eta = 1-\Es{x}\sum_a \bra{\psi} A^x_a \otimes (\hat{A}^x_a)^T \ket{\psi}$ gives $(\eta-\delta)^2 \leq \gamma$ so 
\begin{equation}\label{eq:cc-6}
\eta \leq \sqrt{\gamma} + \delta
\end{equation}
 and 
\begin{align*}
1-\eta &= \Es{x}\sum_a \bra{\psi} A^x_a \otimes (\hat{A}^x_a)^T \ket{\psi}\\
&\leq \Big( \Es{x}\sum_a \bra{\psi} A^x_a \otimes (A^x_a)^T \ket{\psi}\Big)^{1/2}\Big( \Es{x}\sum_a \bra{\psi} \hat{A}^x_a \otimes (\hat{A}^x_a)^T \ket{\psi}\Big)^{1/2}\;.
\end{align*}
using Lemma~\ref{lem:psi-cs}. Thus $(1-\delta)(1-\hat{\delta}) \geq (1-\eta)^2$ which implies $\hat{\delta} \leq 2\eta \leq 2\sqrt{\gamma}$ by~\eqref{eq:cc-6}. Proceeding as for~\eqref{eq:cc-1}, this shows~\eqref{eq:cc-5}. 

Combining~\eqref{eq:cc-2},~\eqref{eq:cc-3} and~\eqref{eq:cc-5} proves the lemma. 
\end{proof}


\subsection{Rounding operators}

We introduce two simple lemma originally due to Connes~\cite{connes1976classification} (who proved them in the much more general setting of semifinite von Neumann algebras). The lemma allow one to provide estimates on $\|f(A)-g(B)\|_F$ when $A,B$ are Hermitian operators on $\mH$ and $f,g$ real-valued functions. As discussed in the introduction these lemma were previously used in~\cite{slofstra2018entanglement} to show a weaker result than we show here (which was sufficient for their purposes). 

For $\lambda\in\R$ define $\chi_{\geq \lambda}: \R\to\R$ by $\chi_{\geq \lambda}(x)=1$ if $x\geq \lambda$ and $0$ otherwise. Extend $\chi_{\geq \lambda}$ to Hermitian operators on $\mH$ using the spectral calculus. 
The first lemma appears as Lemma~5.6 in~\cite{slofstra2018entanglement}. 

\begin{lemma}\label{lem:int}
    Let $\rho$ be a positive semidefinite operator on a finite-dimensional Hilbert space. Then
    \begin{equation*}
        \int_{0}^{+\infty} \chi_{\geq \sqrt{\lambda}}(\rho^{1/2}) d\lambda = \rho\;,
    \end{equation*}
		where the integral is taken with respect to the Lebesgue measure on $\R_+$. 
\end{lemma}

The second lemma appears as Lemma~5.5 in~\cite{slofstra2018entanglement}. 

\begin{lemma}\label{lem:connestrick}
    Let $\rho, \sigma$ be positive semidefinite operators on a finite-dimensional
    Hilbert space. Then 
    \begin{equation*}
        \int_0^{+\infty} \norm{\chi_{\geq \sqrt{\lambda}}(\rho^{1/2}) 
            - \chi_{\geq \sqrt{\lambda}}(\sigma^{1/2})}^2_F d\lambda
            \leq \big\|\sigma^{1/2} - \rho^{1/2}\|_F \big\|\sigma^{1/2} + \rho^{1/2}\|_F\;.
    \end{equation*} 
\end{lemma}

\subsection{Orthonormalization}

The following result shows that an approximately consistent strategy is always close to a projective strategy. The result first appears in~\cite{kempe2011parallel}. The statement that we give here is taken from~\cite{ml2020}. 

\begin{proposition}\label{prop:ortho}
Let $0\leq \delta\leq 1$. Let $\ket{\psi}$ be a state on $\mH \otimes \mH$ whose reduced densities on either subsystem are identical. Let $k$ be an integer and $Q_1,\ldots,Q_k$ positive semidefinite operators on $\mH$ such that $\sum_i Q_i = \Id$. Let 
\begin{equation*} 
\delta \,=\, 1-\sum_i \, \bra{\psi} Q_i \otimes  Q_i^T \ket{\psi} \;.
\end{equation*}
Then there exists orthogonal projections $P_1,\ldots, P_k$ on $\mH$ such that $\sum_i P_i = \Id$ and
\begin{equation}\label{eq:cons-ccl}
\sum_i \bra{\psi} (P_i-Q_i)^2 \otimes \Id \ket{\psi}\,\leq\, O\big(\delta^{1/4}\big)\;.
\end{equation}
\end{proposition}

\section{Main result}
\label{sec:main}

The following is our main result. It states that a strategy that induces a correlation which is almost synchronous must be proportionately close, in a precise sense, to a projective maximally entangled strategy.

\begin{theorem}\label{thm:main}
There are universal constants $c,C>0$ such that the following holds. Let $\mX$ and $\mA$ be finite sets and $\nu$ a distribution on $\mX$. Let $\strategy=(\ket{\psi},A)$ be a symmetric strategy and $\delta = \delta_\sync(\strategy;\nu)$.
 Then there is a measure $\mu$ on $\mathbb{R}_+$ and a family of Hilbert spaces $\mH_\lambda \subseteq \mH$, for $\lambda\in\mathbb{R}_+$ (both depending on $\ket{\psi}$ only) such that the following holds. 
For every $\lambda\in \R_+$ there is a maximally entangled state $\ket{\psi_\lambda}\in \mH_\lambda\otimes \mH_\lambda$ and for each $x$ a projective measurement $\{A^{\lambda,x}_a\}$ on $\mH_\lambda$ such that 
\begin{enumerate}
\item Letting $\rho$ be the reduced density of $\ket{\psi}$ on $\mH$ and $\rho_\lambda$ the totally mixed state on $\mH_\lambda\subseteq\mH$,
\begin{equation}\label{eq:rho-plambda}
 \rho \,=\, \int_\lambda \rho_\lambda d\lambda\;.
\end{equation}
\item The $\{\strategy_\lambda = (\ket{\psi_\lambda},A^\lambda)\}$ provide an approximate decomposition of $\strategy$ as a convex sum of projective maximally entangled (PME) strategies, in the following sense: 
\begin{equation}\label{eq:lambda-strat-0}
\Es{x\sim \nu} \sum_a \int_\lambda \Tr\big( \big(A^x_a - A^{\lambda,x}_a\big)^2 \rho_\lambda\big) d\mu(\lambda)\,\leq\,C \,\delta^c\;.
\end{equation}
\end{enumerate}
\end{theorem}

The key point in Theorem~\ref{thm:main} is that the error estimates are independent of the dimension of $\mH$ and of the size of the sets $\mX$ and $\mA$. We remark that the integral over $\lambda$ can be written as a finite convex sum. This is evident from the definition of $\rho_\lambda$ as a multiple of the projection $P_\lambda$ defined in~\eqref{eq:def-plambda}. Since $\mH$ is finite-dimensional $\rho$ has a discrete spectrum and $P_\lambda$ takes on a finite set of values. 

We note that the theorem does not imply that $\ket{\psi}$ itself is close to a maximally entangled state. Rather,~\eqref{eq:rho-plambda} implies that \emph{after tracing an ancilla}, which contains the index $\lambda$, this is the case. It is not hard to see that this is unavoidable by considering a game such that there exists multiple optimal strategies for the game that are not unitarily equivalent. For example one can consider a linear system game that tests the group generated by the Pauli matrices $\sigma_X,\sigma_Z$ and $\sigma_Y$; this can be obtained from three copies of the Magic Square game as in e.g.~\cite[Appendix A]{coladangelo2019verifier}. This game can be won with probability $1$ using any state of the form
\[ \ket{\psi} \, = \, \ket{\phi^+}_{A_1B_1} \ket{\phi^+}_{A_2 B_2} \big( \alpha \ket{00}_{A_3B_3} + \beta \ket{11}_{A_3B_3} \big)\;,\]
where $\ket{\phi^+}$ is an EPR pair (rank-$2$ maximally entangled state) and the measurement operators are block-diagonal with respect to the third system (i.e.\ $X = X'_{A_1A_2} \otimes \proj{0}_{A_3} +  X''_{A_1A_2} \otimes \proj{1}_{A_3}$ for the first player). Crucially the measurement operator's dependence on the third system cannot be removed by a local unitary, because the $X'$ and $X''$ components are not unitarily related. Although the strategy cannot be locally dilated to a maximally entangled strategy in the sense of Definition~\ref{def:dilation} it is not hard to see that it nevertheless has a decomposition of the form promised by Theorem~\ref{thm:main}. 

\subsection{Corollaries}

Before turning to the proof we give a pair of corollaries. The first shows that the conclusions of the theorem are maintained even without the assumption that $\strategy$ is symmetric. 

\begin{corollary}
Let $\strategy=(\ket{\psi},A,B)$ be a strategy. Let $\nu$ be a distribution on $\mX$ and $\delta = \delta_\sync(\strategy;\nu)$. Then the same conclusions as Theorem~\ref{thm:main} hold (for different constants $c,C)$, where $\rho$ is chosen as the reduced density of $\ket{\psi}$ on either $\mH_A$ (in which case the conclusions apply to $\{A^x_a\}$) or $\mH_B$ (in which case they apply to $\{B^y_b\}$).
\end{corollary}

\begin{proof} 
Using Lemma~\ref{lem:psi-cs} and Jensen's inequality it follows that 
\begin{align*}
\Es{x\sim \nu} \sum_a \bra{\psi} A^x_a \otimes B^x_a \ket{\psi} 
&\leq \sqrt{1-\delta_\sync(\strategy_A;\nu)}\sqrt{1-\delta_\sync(\strategy_B;\nu)}\;,
\end{align*}
where $\strategy_A = (\ket{\psi_A},A)$, $\strategy_B = (\ket{\psi_B},B)$ with $\ket{\psi_A}$ and $\ket{\psi_B}$ canonical purifications of the reduced density of $\ket{\psi}$ on $\mH_A$ and $\mH_B$ respectively. 
This allows us to apply Theorem~\ref{thm:main} separately to each symmetric strategy $\strategy_A$ and $\strategy_B$ to obtain the desired conclusions. 
\end{proof}

The second corollary shows that the conclusions of the theorem imply an approximate decomposition of the correlation implied by $\strategy$ as a convex combination of synchronous correlations. 

\begin{corollary}\label{cor:c2}
Let $\strategy=(\ket{\psi},A,B)$ be a projective strategy. Let $\nu$ be a distribution on $\mX$, $\delta = \delta_\sync(\strategy;\nu)$, and $\{\strategy_\lambda\}$ and $\mu$ the family of strategies and the measure obtained from Theorem~\ref{thm:main}. Let $C$ (resp. $C^\lambda$) be the correlation induced by $\strategy$ (resp. $\strategy_\lambda$). Let $\tilde{\nu}$ be any distribution on $\mX \times \mX$ with marginal $\nu$. Then 
\begin{equation}\label{eq:c2-0}
\Es{(x,y)\sim \tilde{\nu}}\; \sum_{a,b}\; \Big| C_{x,y,a,b} -  \int_\lambda C^\lambda_{x,y,a,b} d\lambda\Big| \,=\, \poly(\delta)\;.
\end{equation}
\end{corollary} 

The assumption that $\strategy$ is projective is without loss of generality since by Lemma~\ref{lem:naimark} any correlation $C \in C_q(\mX,\mA)$ can be achieved by a projective strategy. 
By an averaging argument the corollary immediately implies that for any game $\game$ with question distribution $\tilde{\nu}$ there is a $\lambda$ such that $\strategy_\lambda$ succeeds at least as well as $\strategy$ in $\game$, up to an additive loss $\poly(\delta)$. 

\begin{proof}
Fix $\strategy$, $\nu$, $\tilde{\nu}$ and $\{\strategy^\lambda\}$, $\mu$ as in the statement of the corollary. Conjugating the $B^y_b$ by a unitary if necessary we assume without loss of generality that the reduced densities of $\ket{\psi}$ on either subsystem are identical.
 For every $\lambda$ define a symmetric strategy $\tilde{\strategy}_\lambda = (\ket{\psi_\lambda},A)$ and let $\tilde{C}^\lambda$ be the associated correlation. We first show that 
\begin{equation}\label{eq:c2-1}
 \Es{(x,y)\sim \tilde{\nu}} \sum_{a,b} \big| C_{x,y,a,b} - \int_\lambda\tilde{C}^\lambda_{x,y,a,b}d\lambda\big |  \,=\, \poly(\delta)\;.
\end{equation}
For this we show that 
\begin{equation}\label{eq:c2-1b}
\Es{x,y} \sum_{a,b} \big| C_{x,y,a,b} - \Tr\big(  (A^x_a)^T B^y_b(A^x_a)^T \rho \big)\big|\,=\, O\big(\sqrt{\delta}\big)\;,
\end{equation}
and
\begin{equation}\label{eq:c2-1c}
\int_\lambda \Es{x,y} \sum_{a,b} \big| \tilde{C}^\lambda_{x,y,a,b} - \Tr\big( (A^x_a)^T B^y_b (A^x_a)^T\rho_\lambda \big)\big|\,=\, O\big({\delta}^{c/4}\big)\;.
\end{equation}
Together with~\eqref{eq:rho-plambda}, combining~\eqref{eq:c2-1b} and~\eqref{eq:c2-1c} through the triangle inequality gives~\eqref{eq:c2-1}. To show~\eqref{eq:c2-1b}, write using the triangle inequality
\begin{align}
\Es{x,y} \sum_{a,b} \big| C_{x,y,a,b} - \Tr\big( (A^x_a)^T B^y_b (A^x_a)^T\rho \big)\big|
&\leq \Es{x,y} \sum_{a,b} \big| \bra{\psi} \big( \Id \otimes(A^x_a)^T- A^x_a\otimes \Id\big) B^y_b (\Id \otimes (A^x_a)^T )\ket{\psi} \big| \label{eq:c2-2a}\\
&\quad + \Es{x,y} \sum_{a,b} \bra{\psi} \big(\Id\otimes (A^x_a)^T \big)B^y_b \big( \Id \otimes (A^x_a)^T- A^x_a \otimes \Id\big)\ket{\psi}\big|\;,\label{eq:c2-2b}
\end{align}
where we used the assumption that each $B^y_b$ is a projection. 
Each of the two terms on the right-hand side is bounded in the same manner. We show how to bound the first: 
\begin{align}
\Es{x,y} \sum_{a,b} \big| \bra{\psi} \big( \Id \otimes(A^x_a)^T &- A^x_a\otimes \Id\big) B^y_b (\Id \otimes (A^x_a)^T )\ket{\psi} \big|\notag\\
&\leq \Big(\Es{x,y} \sum_{a,b}  \bra{\psi} \big( \Id \otimes (A^x_a)^T - A^x_a\otimes \Id\big) B^y_b \big( \Id \otimes (A^x_a)^T - A^x_a\otimes \Id\big)\ket{\psi} \Big)^{1/2}\notag\\
&\qquad\qquad\cdot \Big(\Es{x,y} \sum_{a,b}  \bra{\psi} ( \Id \otimes (A^x_a)^T) B^y_b ( \Id \otimes(A^x_a)^T )\ket{\psi} \Big)^{1/2} \notag\\
&\leq \Big( \Es{x} \sum_a \bra{\psi}\big( \Id \otimes (A^x_a)^T- A^x_a\otimes \Id\big)^2\ket{\psi} \Big)^{1/2} \cdot 1\label{eq:c2-2e}\\
&\leq \sqrt{2\delta}\;,\notag
\end{align}
where the first inequality is Cauchy-Schwarz, the second uses $\sum_a A^x_a = \sum_b (B^y_b)^2 = \Id$, and the last follows by expanding the square and using the definition of $\delta$. This bounds~\eqref{eq:c2-2a}. Together with a similar bound for the term in~\eqref{eq:c2-2b} this shows~\eqref{eq:c2-1b}. To show~\eqref{eq:c2-1c}, we proceed similarly up until the last step~\eqref{eq:c2-2e}, at which point a bound on
\begin{equation}\label{eq:c2-2e1}
\int_\lambda \Es{x} \sum_a \bra{\psi_\lambda}\big( \Id \otimes (A^x_a)^T- A^x_a\otimes \Id\big)^2\ket{\psi_\lambda}
\end{equation}
is required. To obtain this, we first note that 
\begin{equation}\label{eq:c2-2e2}
\int_\lambda \Es{x} \sum_a \bra{\psi_\lambda}\big( \Id \otimes (A^{\lambda,x}_a)^T- A^{\lambda,x}_a\otimes \Id\big)^2\ket{\psi_\lambda}\,=\,0\;,
\end{equation}
because $A^{\lambda,x}_a$ is supported on the support of $\rho_\lambda$, which is totally mixed on its support. 
Moreover, forming the difference we have (using that $\{A^{\lambda,x}_a\}$ and $\{A^x_a\}$ are projective)
\begin{align*}
|\eqref{eq:c2-2e1}-\eqref{eq:c2-2e2} |
&= 2 \int_\lambda \Es{x} \sum_a \Big(\bra{\psi_\lambda} (A^{\lambda,x}_a - A^{x}_a)\otimes (A^{\lambda,x}_a)^T\ket{\psi_\lambda} + \bra{\psi_\lambda} A^{x}_a \otimes (A^{\lambda,x}_a- A^{x}_a)^T\otimes \ket{\psi_\lambda} \Big)\;.
\end{align*}
Each of the two terms on the right-hand side is bounded by an application of the Cauchy-Schwarz inequality followed by~\eqref{eq:lambda-strat-0}. This shows~\eqref{eq:c2-1b}, and hence~\eqref{eq:c2-1}.

Having established~\eqref{eq:c2-1} we now prove
\begin{equation}\label{eq:c2-3}
  \int_{\lambda}\Es{(x,y)\sim \tilde{\nu}} \sum_{a,b} \big| \tilde{C}^\lambda_{x,y,a,b} -{C}^\lambda_{x,y,a,b}\big | d\lambda \,=\, \poly(\delta)\;.
\end{equation}
Combining~\eqref{eq:c2-1} and~\eqref{eq:c2-3} shows~\eqref{eq:c2-0}, concluding the proof. To show~\eqref{eq:c2-3} we apply Lemma~\ref{lem:close-cor} for each $\lambda$ to the strategy $\strategy = (\ket{\psi_\lambda},A^\lambda,B)$ here and $\hat{A}$ in Lemma~\ref{lem:close-cor} is $A$ here. Since $\ket{\psi_\lambda}$ is maximally entangled and $A^\lambda$ supported on its support, $\delta$ in Lemma~\ref{lem:close-cor} equals $0$. Applying the lemma followed by Jensen's inequality gives 
\begin{align*}
 \int_{\lambda}\Es{(x,y)\sim \tilde{\nu}} \sum_{a,b} \big| \tilde{C}^\lambda_{x,y,a,b} -{C}^\lambda_{x,y,a,b}\big | d\lambda
&= O\Big(\Big( \int_\lambda \Es{x} \sum_a \Tr\big( \big( A^{\lambda,x}_a - A^x_a\big)^2 \rho_\lambda \big)d_\lambda  \Big)^{1/2}\Big)\\
&= \poly(\delta)\;,
\end{align*}
by~\eqref{eq:lambda-strat-0}. This shows~\eqref{eq:c2-3} and concludes the proof. 
 \end{proof}

\subsection{Proof of Theorem~\ref{thm:main}}

We now prove the theorem. As in the statement of Theorem~\ref{thm:main}, let $\strategy = (\ket{\psi},A)$ be a symmetric strategy. Let $\rho$ be the reduced density of $\ket{\psi}$ on $\mH$. 
As a first step in the proof we apply Proposition~\ref{prop:ortho} to obtain a nearby symmetric projective strategy with nearly the same success probability. 

\begin{lemma}\label{lem:main-step}
There is a projective symmetric strategy $\strategy'= (\ket{\psi},B)$ such that letting $\delta'=\delta_\sync(\strategy',\nu)$ 
then $\delta'= O(\delta^{1/8})$ and 
\begin{equation}\label{eq:main-step-0}
\Es{x\sim\nu} \sum_a \Tr\big( (A^x_a - B^x_a)^2 \rho\big) \,=\, O\big(\delta^{1/4}\big)\;.
\end{equation}
\end{lemma}

\begin{proof}
For each $x$ let $\delta_x = 1-\sum_a \bra{\psi} A^x_a \otimes (A^x_a)^T \ket{\psi}$. By definition of $\delta$ it holds that 
\begin{equation}\label{eq:main-step-1a}
\delta \,=\, \Es{x\sim \nu} \delta_x\;.
\end{equation}
 For each $x\in \mX$, applying Proposition~\ref{prop:ortho} to the measurement $\{A^x_a\}$ gives a projective measurement $\{B^x_a\}$ such that 
\[\sum_a \Tr\big( (A^x_a - B^x_a)^2 \rho\big) \,=\, O\big(\delta_x^{1/4}\big)\;.\]
Taking the expectation over $x$,
\begin{align}
\Es{x\sim\nu} \sum_a \Tr\big( (A^x_a - B^x_a)^2 \rho\big) &= O\Big( \Es{x\sim\nu} \delta_x^{1/4}\Big)\notag\\
&= O\Big( \Big( \Es{x\sim\nu} \delta_x \Big)^{1/4}\Big)\notag\\
&= O\big(\delta^{1/4}\big)\;,\label{eq:main-step-1}
\end{align}
where the second line uses Jensen's inequality and the third uses~\eqref{eq:main-step-1a}. This gives~\eqref{eq:main-step-0}. Let $\delta'=\delta_\sync(\strategy',\nu)$. Then 
\begin{align*}
\delta'-\delta &= \Es{x}\sum_a \big( \bra{\psi} A^x_a \otimes(A^x_a)^T \ket{\psi} - \bra{\psi} B^x_a \otimes (B^x_a)^T \ket{\psi} \big)\\
&= \Es{x}\sum_a  \big(\bra{\psi} (A^x_a - B^x_a) \otimes (A^x_a)^T \ket{\psi} + \bra{\psi} B^x_a \otimes (A^x_a-B^x_a)^T\ket{\psi} \big)\\
&\leq \Big(\Es{x}\sum_a \Tr\big((A^x_a - B^x_a)^2\rho\big)\Big)^{1/2} \Big( \Big( \Es{x}\sum_a \Tr( (A^x_a)^2 \rho )\Big)^{1/2} 
+  \Big( \Es{x}\sum_a \Tr( (B^x_a)^2 \rho )\Big)^{1/2}\Big)\\
&\leq   O\big(\delta^{1/8}\big)\cdot \sqrt{2}\;,
\end{align*}
where the second inequality follows from the Cauchy-Schwarz inequality and the last uses~\eqref{eq:main-step-1} to bound the first term and that for each $x$, $\{A^x_a\}_a$ and $\{B^x_a\}$ are measurements.  This shows $\delta'= O(\delta^{1/8})$, as claimed. 
\end{proof}

For every $\lambda \in \R_+$ let 
\begin{equation}\label{eq:def-plambda}
P_\lambda \,=\, \chi_{\geq \lambda}(\rho)
\end{equation} 
be the projection on the direct sum of all eigenspaces of $\rho$ with associated eigenvalue at least $\lambda$. Using Lemma~\ref{lem:int}, $\int_\lambda \Tr(P_\lambda)d\lambda = 1$, so $d\mu(\lambda)=\Tr(P_\lambda)d\lambda$ is a probability measure. Let $\mH_\lambda$ be a Hilbert space of dimension the rank of $P_\lambda$. We endow each $\mH_\lambda$ with an orthonormal basis of eigenvectors of $\rho$ that allows us to view $\mH_\lambda$ as a subspace of $\mH_{\lambda'}$ for any $\lambda' \leq \lambda$, with $\mH_\lambda =\{0\}$ for any $\lambda > \|\rho\|$ and the convention $\mH_0=\mH$.

The next lemma shows a form of approximate commutation between the $\{B^x_a\}$ and $\{P_\lambda\}$. 

\begin{lemma}\label{lem:lambda-strat-1}
The following holds:
\begin{equation}\label{eq:lambda-strat-1a}
\int_\lambda \Es{x\sim \nu} \sum_{a} \big\| {B}_{a}^{x} P_\lambda -  P_\lambda {B}_{a}^{x} \big\|_F^2 \,\leq\, 2\sqrt{2\delta'}\;,
\end{equation}
where $\delta' = O(\delta^{1/8})$ is as in Lemma~\ref{lem:main-step}.
\end{lemma}

\begin{proof}
For convenience in the proof of the lemma 
we identify the set $\mA$ with $\mathbb{Z}_m$, for some integer $m$. Define a family of unitaries $\{U^x_b\}$ indexed by $x\in\mX$ and $b\in \mA$ by 
\begin{equation}\label{eq:def-ux}
U^x_b \,=\, \sum_a e^{2i\pi ab/m} B^x_a\;.
\end{equation}
With this definition we observe that
\begin{align}
\Es{x} \Es{b} \big\| U^x_b \rho^{1/2} - \rho^{1/2} U^x_b \big\|_F^2 
&= 2 - 2 \Es{x}\sum_a \Tr(B^x_a \rho^{1/2} B^x_a \rho^{1/2})\notag\\
&= 2\delta'\;,\label{eq:game-succ-0}
\end{align}
where the expectation over $x$ is taken with respect to the (marginal of) the game distribution $\nu$, the expectation over $b$ is uniform over $\Z_m$, the first equality uses the equality $\Es{b} e^{2i\pi(a-a')/m} = \delta_{a,a'}$ (the Kronecker $\delta$) for all $a,a'\in\mathbb{Z}_m$ and the fact that for every $x$, $\{B^x_a\}_a$ is projective, and the second uses the identity~\eqref{eq:ando}.

For each $x\in \mX$ and $b\in\mB$ let $\sigma^x_b = (U_b^x)^\dagger \rho {U}_{b}^{x} $. Observe that for any $\lambda \in \R_+$,
\[ \chi_{\geq \sqrt{\lambda}} \big((\sigma^x_b)^{1/2}\big) \,=\,(U_b^x)^\dagger \,\chi_{\geq \sqrt{\lambda}}\big(\rho^{1/2}\big)\, {U}_{b}^{x}\;.\]
Hence using the definition~\eqref{eq:def-plambda} of $P_\lambda$, 
\begin{align*}
\Es{x} \Es{b}\,\int_\lambda  \big\|  P_\lambda - ({U}_{b}^{x})^\dagger P_\lambda {U}_{b}^{x} \big\|_F^2 
&= \Es{x} \Es{b}\,\int_\lambda  \big\|  \chi_{\geq \sqrt{\lambda}} \big(\rho^{1/2}\big)- \chi_{\geq \sqrt{\lambda}} \big((\sigma^x_b)^{1/2}\big) \big\|_F^2 \\
&\leq \Es{x} \Es{b} \big\|  \rho^{1/2} -  (U_b^x)^\dagger \rho^{1/2} {U}_{b}^{x} \big\|_F \big\|\rho^{1/2} +  (U_b^x)^\dagger \rho^{1/2} {U}_{b}^{x} \big\|_F\\
&\leq \Big(\Es{x} \Es{b} \big\|  \rho^{1/2} -  (U_b^x)^\dagger \rho^{1/2} {U}_{b}^{x} \big\|_F^2\Big)^{1/2}\Big(\Es{x}\Es{b} \big\|\rho^{1/2} +  (U_b^x)^\dagger \rho^{1/2} {U}_{b}^{x} \big\|_F^2\Big)^{1/2}\\
&\leq\sqrt{2\delta'}\sqrt{4}\;,
\end{align*}
where the inequality on the second line follows from applying
Lemma~\ref{lem:connestrick} independently for each $x$ and $b$, the third line is the Cauchy-Schwarz inequality, and 
for the last inequality the first term is bounded using~\eqref{eq:game-succ-0} and the second using $\|\rho^{1/2}\|_F^2=1$. The claim follows since from the definition~\eqref{eq:def-ux} we get by expanding the left-hand side that for each $x$ and $\lambda$,
\[ \Es{b}\,\big\| U^x_b P_\lambda -  P_\lambda {U}_{b}^{x} \big\|_F^2 
\,=\, \sum_a \big\| B^x_a P_\lambda -  P_\lambda B^x_a \big\|_F^2 \;.\]
\end{proof}

The preceding two lemma in hand, we are ready to give the proof of Theorem~\ref{thm:main}.

\begin{proof}[Proof of Theorem~\ref{thm:main}]
 Fix a symmetric strategy $\strategy = (\ket{\psi},A)$ for $\game$. Let $\{B^x_a\}$ be the family of projective measurements obtained in Lemma~\ref{lem:main-step}, $\strategy'= (\ket{\psi},B)$ and $\delta' = \delta_\sync(\strategy';\nu)$. 

For $\lambda \in \R_+$ let $\tilde{A}_{a}^{\lambda,x} = P_\lambda B_a^x P_{\lambda}$ and $\ket{\psi_\lambda}$ denote the maximally entangled state on $\mH_\lambda\otimes \mH_\lambda$. Then $\tilde{\strategy}_\lambda = (\ket{\psi_\lambda},\tilde{A}^\lambda)$ is a well-defined symmetric strategy. 
 Lemma~\ref{lem:lambda-strat-1} allows us to bound
\begin{align}
\int_\lambda \Es{x} \sum_a \big\| ( B^x_a - \tilde{A}^{\lambda,x}_a)^2 P_\lambda  \big\|_F^2 d\lambda
&=\int_\lambda \Es{x} \sum_a \Tr\big( B^x_a P_\lambda B^x_a (\Id-P_\lambda)\big) d\lambda\notag\\
&=\int_\lambda \Es{x} \sum_a \Tr\big( [B^x_a, P_\lambda][B^x_a, P_\lambda]^\dagger\big) d\lambda\notag\\
&= O\big(\sqrt{\delta'}\big) \;,\label{eq:lambda-strat-2}
\end{align}
where for the rewriting in the first and second lines we used the definition of $\tilde{A}^{\lambda,x}_a$, the fact that $P_\lambda$ is a projection for each $\lambda$, and that $\{B^x_a\}$ is a projective measurement for all $x$, and for the last line we used~\eqref{eq:lambda-strat-1a}.

It remains to turn the strategies $\tilde{\strategy}_\lambda$ into projective strategies. 
For this we apply Proposition~\ref{prop:ortho} to each measurement $\{A^{\lambda,x}_a\}_a$, for all $x$ and $\lambda$. To justify this application we evaluate
\begin{align}
\int_\lambda \Es{x}\sum_a \bra{\psi_\lambda} \tilde{A}^{\lambda,x}_a \otimes \tilde{A}^{\lambda,x}_a \ket{\psi_\lambda} d\mu(\lambda)
&= \int_\lambda \Es{x}\sum_a \Tr\big( B^x_a P_\lambda B^x_a P_\lambda\big)d\lambda\notag\\
&= 1-\frac{1}{2} \int_\lambda \Es{x}\sum_a \big\| B^x_a P_\lambda - P_\lambda B^x_a\big\|_F^2 d\lambda\notag\\
&\geq 1- O\big(\sqrt{\delta'}\big)\;,\label{eq:lambdacons}
\end{align}
where the first equality uses the definition of $d\mu(\lambda)$ and $\tilde{A}^{\lambda,x}_a$ and~\eqref{eq:ando}, the second uses that $B^x_a$ is projective, and the last line is by~\eqref{eq:lambda-strat-1a}.
For each $x$ and $\lambda$ let $\{A^{\lambda,x}_a\}$ be the projective measurement that is associated to $\{\tilde{A}^{\lambda,x}_a\}$ by Proposition~\ref{prop:ortho}. Using Jensen's inequality and~\eqref{eq:lambdacons} the proposition gives the guarantee
\begin{equation}\label{eq:lambda-strat-5}
\int_\lambda \Es{x}\sum_a \Tr\big( \big( A^{\lambda,x}_a - \tilde{A}^{\lambda,x}_a\big)^2 P_\lambda\big) d\lambda \,=\, O\big((\delta')^{1/8}\big)\;.
\end{equation}
For each $\lambda$ the strategy $\strategy_\lambda = (\ket{\psi_\lambda},A^\lambda)$ is a PME strategy by definition, and~\eqref{eq:lambda-strat-0} follows by combining~\eqref{eq:main-step-0},~\eqref{eq:lambda-strat-2} and~\eqref{eq:lambda-strat-5}. 
\end{proof}

\section{Applications to nonlocal games}

We give two applications of Theorem~\ref{thm:main}. The first is to transferring ``rigidity'' statements obtained for PME strategies to the general case. The second is to the class of projection games. 

\subsection{Application to rigidity}
\label{sec:rigid}

As mentioned in the introduction, Theorem~\ref{thm:main} allows one to transfer rigidity statements shown for PME strategies to general strategies. We do not have a general all-purpose statement demonstrating this. Instead we give two simple corollaries that are meant to describe sample applications. The first corollary considers a situation important in complexity theory, where one aims to show that a large family of measurements that constitute a successful strategy in a certain game must in some sense be consistent with a single larger measurement that ``explains'' it; see Subsection~\ref{sec:rigid-1}. The second corollary considers a typical midpoint in a proof of rigidity, where one uses the game condition to derive certain algebraic relations on the measurements that constitute a successful strategy, which are then shown to impose a further structure; see Subsection~\ref{sec:rigid-2}. 

\subsubsection{Application to showing classical soundness}
\label{sec:rigid-1}

Our first application arises in complexity theory when one is trying to show that quantum strategies in a certain nonlocal game obey a certain ``global'' structure. We first state the corollary and then describe a typical application of it. 

\begin{corollary}\label{cor:main}
Let $\game = (\mX,\mA,\nu,D)$ be a symmetric game. Suppose given the following:
\begin{itemize}
\item Finite sets $\mY$ and $\mB$,
\item A joint distribution $p$ on $\mX\times \mY$,
\item For every $(x,y)\in \mX\times \mY$ a function $g_{xy}:\mA\to 2^\mB$, the collection of subsets of $\mB$, such that for any fixed $(x,y)$ the sets $g_{xy}(a)$ for $a\in \mA$ are pairwise disjoint,
\item A convex monotone non-decreasing function $\kappa:[0,1]\to\R_+$,
\end{itemize}
and suppose that given this data the following statement holds:
\begin{quote}
For every $\omega\in[0,1]$ and symmetric PME strategy $\strategy=(\ket{\psi},A)$ that succeeds with probability $\omega$ in $\game$  there is a family of measurements $\{M^y_b\}$ on $\mH$, indexed by $y\in \mY$ and with outcomes $b\in \mB$, such that 
\begin{equation}\label{eq:main-0a}
 \Es{(x,y)\sim p} \sum_a \bra{\psi} A^x_a \otimes M^y_{[g_{xy}(a)]} \ket{\psi} \geq \kappa(\omega)\;,
\end{equation}
where $M^y_{[g_{xy}(a)]} = \sum_{b\in g_{xy}(a)} M^y_b$. 
\end{quote}
Then the same statement extends to arbitrary symmetric projective strategies $\strategy'=(\ket{\psi},A)$, with the right-hand side in~\eqref{eq:main-0a} replaced by $\kappa(\omega-\poly(\delta)) - \poly(\delta)$ where $\delta = \delta_\sync(\strategy';\nu)$.
\end{corollary}

Note that using Lemma~\ref{lem:am-cons} the guarantee~\eqref{eq:main-0a} can equivalently be expressed in terms of a state-dependent distance between $\{A^x_a\}$ and $\{M^x_a = \Es{y\sim p_x} M^y_{[g_{xy}(a)]}\}$, with $p_x$ the conditional distribution $p(x,\cdot)/p(x)$.  The condition that the strategy $\strategy'$ should be  symmetric projective is very mild, as projectivity can always be obtained by applying Naimark dilation (Lemma~\ref{lem:naimark}) and symmetry is generally obtained as a consequence of symmetry in the game. 

The loss in quality of approximation guaranteed by the corollary depends polynomially on $\delta_\sync(\strategy';\nu)$. In many cases this quantity can be bounded directly from a high success probability in the game. This is the case if for example the distribution $\nu$ is such that $\nu(x,x) \geq c\nu(x)$ for some $c>0$ and all $x$, where recall that by slight abuse of notation we use $\nu(\cdot)$ to denote the marginal on either player. In this case any strategy such that $\omega_q(\game;\strategy)\geq 1-\eps$ has $\delta_\sync(\strategy;\nu_A)\leq \eps/c$ and so no further assumption is necessary. 

The assumption made in the proposition is typical of a rigidity result and is specifically meant to illustrate the potential applicability of our result to a setting such as that of the low-individual degree test of~\cite{ml2020}, which forces successful strategies in a certain game to necessarily have a specific ``global'' structure.
For purposes of illustration we state an over-simplified version of the main result from~\cite{ml2020} result as follows. 

\begin{theorem}[Theorem 1.3 in~\cite{ml2020}, informal]\label{thm:cld}
Suppose that a symmetric strategy $\strategy = (\ket{\psi},A)$ succeeds in the ``degree-$d$ low individual degree game'' $\game_{ld}$, which has $\mX = \F_q^m$ and $\mA = \F_q$, with probability at least $1-\eps$.
Then there exists a projective measurement $G = \{G_g\}$
whose outcomes~$g$ are $m$-variate polynomials over $\F_q$ of individual degree at most $d$
such that
\begin{equation}\label{eq:cld}
\Es{x \sim \F_q^m}\sum_{a \in \F_q} \sum_{g:\,g(x) = a} \bra{\psi} A^{x}_{a} \otimes G_g \ket{\psi}
\,\geq\, 1 - \poly(m) \cdot (\poly(\eps) + \poly(d/q))\;.
\end{equation}
\end{theorem}

To apply Corollary~\ref{cor:main} to the setting of Theorem~\ref{thm:cld}, let the game $\game$ in Corollary~\ref{cor:main} be the `degree-$d$ low individual degree game'' $\game_{ld}$ from Theorem~\ref{thm:cld}. Let $\mY=\{y\}$ be a singleton and $\mA$ be the set of  $m$-variate polynomials over $\F_q$ of individual degree at most $d$. Let $p$ be uniform over $\F_q^m \times \mY$. For every $x\in \F_q^m$ and $a \in \F_q$ let $g_{xy}(a)$ be the collection of polynomials that evaluate to $a$ at $x$. Then~\eqref{eq:cld} gives~\eqref{eq:main-0a} with $\kappa(\omega) = \poly(m) \cdot (\poly(\eps) + \poly(d/q))$ where $\eps = 1-\omega$. To conclude we note that for the specific game $\game_{ld}$ the condition $\nu(x,x) \geq c\nu(x)$ for some $c>0$ mentioned earlier holds, which allows us to bound $\delta_\sync$ by $O(\eps)$.\footnote{In fact, for this game $c$ is only inverse polynomial in $m$, which still suffices given the form of $\kappa$.} In conclusion, Corollary~\ref{cor:main} shows that to prove Theorem~\ref{thm:cld}, provided one is willing to accept a small loss in the approximation quality it is sufficient to prove it for PME strategies. As observed in the introduction, this allows for a significant simplification in the technical steps of the proof. 

We give the proof of the corollary. 

\begin{proof}[Proof of Corollary~\ref{cor:main}]
Fix a symmetric projective strategy $\strategy' = (\ket{\psi},A)$ in $\game$ and let $\omega$ denote its success probability. 
For each $\lambda \in \R_+$ let $\strategy_\lambda = (\ket{\psi_\lambda},A^\lambda)$ be the PME strategy promised by Theorem~\ref{thm:main} and $\omega_\lambda$ its probability of success in $\game$. Let $\{M^{\lambda,y}_b\}$ be the family of measurements promised by the assumption of Corollary~\ref{cor:main}, i.e. such that 
\[ \Es{(x,y)\sim p} \sum_a  \bra{\psi_\lambda} A^{x,\lambda} \otimes   M^{\lambda,y}_{[g_{xy}(a)]} \ket{\psi_\lambda}\, \geq\, \kappa(\omega_\lambda)\;.\]
Averaging with respect to the probability measure with density $d\mu(\lambda)$ and using that $\kappa$ is assumed to be convex monotone it follows that 
\begin{align}\label{eq:main-2}
 \int_\lambda \Es{x,y} \sum_a  \bra{\psi_\lambda} A^{x,\lambda}_a \otimes  M^{\lambda,y}_{[g_{xy}(a)]} \ket{\psi_\lambda} d\mu(\lambda) 
&\geq \kappa\Big(\int_\lambda \omega_\lambda d\mu(\lambda)\Big)\notag\\
&\geq \kappa(\omega')\;,
\end{align}
where $\omega' = \omega-\poly(\delta)$ by Corollary~\ref{cor:c2}.

\begin{claim}
The following holds:
\begin{equation}\label{eq:main-3}
 \int_\lambda \Es{x,y} \sum_a \bra{\psi_\lambda} A^{x}_a \otimes  M^{\lambda,y}_{[g_{xy}(a)]} \ket{\psi_\lambda}  d\mu(\lambda)\,\geq\, \kappa(\omega') - \poly(\delta)\;.
\end{equation}
\end{claim}

\begin{proof}
For any $\lambda\in \R_+$ we have 
\begin{align}
\big| \bra{\psi_\lambda} \big(A^{x}_a - A^{x,\lambda}_a\big)&\otimes  M^{\lambda,y}_{[g_{xy}(a)]} \ket{\psi_\lambda}\big|\notag\\
&\leq \big| \bra{\psi_\lambda} \big(A^{x}_a - A^{x,\lambda}_a\big) A^x_a \otimes  M^{\lambda,y}_{[g_{xy}(a)]} \ket{\psi_\lambda}\big|
 + \big| \bra{\psi_\lambda}  A^{x,\lambda}_a\big(A^{x}_a - A^{x,\lambda}_a\big)  \otimes  M^{\lambda,y}_{[g_{xy}(a)]} \ket{\psi_\lambda}\big|\notag\\
&\leq \big\| \big(A^{x}_a - A^{x,\lambda}_a\big)\otimes\Id\ket{\psi_\lambda}\big\|\big( \big\| A^{x,\lambda}_a\otimes \Id \ket{\psi_\lambda}\big\|+\big\| A^{x}_a\otimes \Id \ket{\psi_\lambda}\big\|\big)\;,\label{eq:b-1}
\end{align}
where the first inequality uses that both $A^x_a$ and $A^{\lambda,x}_a$ are projections and the second inequality uses $\|M^{\lambda,y}_{[g_{xy}(a)]}\|\leq 1$ for all $x,y$ and $a$. Averaging over $\lambda$,
\begin{align*}
 \int_\lambda \Big|\Es{x,y} \sum_a & \bra{\psi_\lambda} \big(A^{x}_a - A^{\lambda,x}_a\big) \otimes  M^{\lambda,y}_{[g_{xy}(a)]} \ket{\psi_\lambda}\Big| d\mu(\lambda)\\
&\leq \int_\lambda \Es{x}\sum_a \big\| \big(A^{x}_a - A^{x,\lambda}_a\big)\otimes \Id\ket{\psi_\lambda}\big\|\big( \big\| A^{x,\lambda}_a\otimes \Id \ket{\psi_\lambda}\big\|+\big\| A^{x}_a\otimes \Id \ket{\psi_\lambda}\big\|\big) d\mu(\lambda)\\
&\leq \Big( \int_\lambda \Es{x}\sum_a \big\| \big(A^{x}_a - A^{x,\lambda}_a\big)\otimes \Id\ket{\psi_\lambda}\big\|^2  d\mu(\lambda)\Big)^{1/2} \\
&\qquad\cdot \Big( \int_\lambda \Es{x}\sum_a \big( \big\| A^{x,\lambda}_a\otimes \Id \ket{\psi_\lambda}\big\|+\big\| A^{x}_a\otimes \Id \ket{\psi_\lambda}\big\|\big)^2  d\mu(\lambda)\Big)^{1/2}\\
&\leq \poly(\delta)\;,
\end{align*}
where the first inequality uses~\eqref{eq:b-1}, the second is the Cauchy-Schwarz inequality, and the last uses~\eqref{eq:lambda-strat-0} to bound the first term, since for every $\lambda$, $x$ and $a$,
\[\big\| \big(A^{x}_a - A^{x,\lambda}_a\big)\otimes \Id\ket{\psi_\lambda}\big\|^2 
\,=\, \Tr\big(\big(A^{x}_a - A^{x,\lambda}_a\big)^2 \rho_\lambda\big) \;.\]
\end{proof}

For each $y,b$ define 
\[ M^y_b \,=\, \rho^{-1/2}\Big(\int_\lambda \frac{1}{\Tr(P_\lambda)} P_\lambda M^{\lambda,y}_b P_\lambda d\mu(\lambda) \Big)\rho^{-1/2}\;,\]
and note that $M^y_b \geq 0$ and 
\begin{equation}
 \sum_b M^y_b \,=\, \rho^{-1/2}\Big(\int_\lambda \frac{1}{\Tr(\lambda)} P_\lambda  d\mu(\lambda) \Big)\rho^{-1/2} \,=\, \Id\;,
\end{equation}
by~\eqref{eq:rho-plambda}. Thus for each $y$, $\{M^y_b\}$ is a valid measurement. Moreover, using~\eqref{eq:ando} we get 
 \begin{align*}
  \Es{x,y} \sum_a \bra{\psi} A^{x}_a  \otimes  M^{y}_{[g_{xy}(a)]} \ket{\psi}
&=  \Es{x,y} \sum_a \Tr\big( A^{x}_a  \rho^{1/2}  M^{y}_{[g_{xy}(a)]} \rho^{1/2}\big)\\
&=  \int_\lambda \Es{x,y} \sum_a \Tr\big( A^{x}_a  P_\lambda  M^{\lambda,y}_{[g_{xy}(a)]} P_\lambda \big) d\lambda\\
&=  \int_\lambda \Es{x,y} \sum_a \bra{\psi_\lambda} A^{x}_a  \otimes  M^{\lambda,y}_{[g_{xy}(a)]} \ket{\psi_\lambda} d\mu(\lambda)\\
 &\geq \kappa(\omega') - \poly(\delta)\;,
\end{align*}
where the last inequality is by~\eqref{eq:main-3}.
\end{proof}

\subsubsection{Application to showing algebraic relations}
\label{sec:rigid-2}

Our second application concerns rigidity statements that go through algebraic relations, as is exemplified by the rigidity proofs for games such as the CHSH game, the Magic Square game, as well as more general classes of games; see e.g.~\cite{coladangelo2017robust} for an exposition of this approach.

\begin{corollary}\label{cor:main2}
Let $\game = (\mX,\mA,\nu,D)$ be a symmetric nonlocal game. Suppose that $\{0,1\}\subseteq\mX$ and for any symmetric projective strategy $\strategy = (\ket{\psi},A)$ in $\game$, for $x\in\{0,1\}$, $\{A^x_0,A^x_1\}$ is a two-outcome measurement that can be represented as an observable $A^x = A^x_0-A^x_1$. 
Suppose that the following statement holds for some concave monotone non-decreasing function $\kappa:[0,1]\to\R_+$.
\begin{quote}
For every $\omega\in[0,1]$ and symmetric PME strategy $\strategy=(\ket{\psi},A)$ that succeeds with probability $\omega$ in $\game$  it holds that 
\begin{equation}\label{eq:main-0b}
 \Tr \big( \big( A^0 A^1-A^1A^0\big)^2 \rho\big)\leq \kappa(1-\omega)\;.
\end{equation}
\end{quote}
Then the same statement extends to arbitrary symmetric projective strategies $\strategy'=(\ket{\psi},A)$, with the right-hand side in~\eqref{eq:main-0b} replaced by $\kappa(\omega+\poly(\delta)) + \poly(\delta)$ where 
\begin{equation}\label{eq:delta-rigid}
\delta \,=\, \max \big\{ \delta_\sync(\strategy';q)\,,\; \delta_\sync(\strategy';\nu)\big\}
\end{equation}
 with $q$ the uniform distribution on $\{0,1\}\subseteq \mX$ and $\nu$ the marginal of the game distribution on $\mX$.
\end{corollary}

Since the aim of the corollary is to give a ``toy'' application of our results we sketch the proof but omit the details. 

\begin{proof}[Proof sketch]
Fix a symmetric projective strategy $\strategy' = (\ket{\psi},A)$ in $\game$ and let $\omega$ denote its success probability. 
For each $\lambda \in \R_+$ let $\strategy_\lambda = (\ket{\psi_\lambda},A^\lambda)$ be the PME strategy promised by Theorem~\ref{thm:main} and $\omega_\lambda$ its probability of success in $\game$. Further let $\tilde{\strategy}_\lambda = (\ket{\psi_\lambda},A)$.
First we claim that by an argument similar to the derivation of~\eqref{eq:c2-1} in the proof of Corollary~\ref{cor:c2} it holds that 
\begin{equation}\label{eq:rigid-rel-1}
 \int_\lambda\delta_{\sync}(\tilde{\strategy}_\lambda;q)d\lambda \,=\, \poly(\delta)\;,
\end{equation}
where to show this we use that the definition of $\delta$ in~\eqref{eq:delta-rigid} involves measuring almost synchronicity under $q$. 
We may then achieve the desired conclusion as follows. First we note that 
\begin{equation}\label{eq:rigid-rel-2}
 \Tr \big( \big( A^0 A^1-A^1A^0\big)^2 \rho\big)\,=\, \int_\lambda \Tr \big( \big( A^0 A^1-A^1A^0\big)^2 \rho_\lambda\big)d\lambda
\end{equation}
by~\eqref{eq:rho-plambda}. Next we use~\eqref{eq:rigid-rel-1} to show
\begin{equation}\label{eq:rigid-rel-3}
 \int_\lambda \Big| \Tr \big( \big( A^0 A^1-A^1A^0\big)^2 \rho_\lambda\big)\,-\, \Tr \big( \big( A^{\lambda,0} A^{\lambda,1}-A^{\lambda,1}A^{\lambda,0}\big)^2 \rho_\lambda\big)\Big| \,=\, \poly(\delta)\;,
\end{equation}
where~\eqref{eq:lambda-strat-0} is used, informally, to ``switch'' operators from one side of the tensor product to the other so that~\eqref{eq:rigid-rel-1} can be applied to each operator in an expansion of the square in turn. Finally, the second term on the left-hand side in~\eqref{eq:rigid-rel-3}  is at most $\kappa(\omega-\poly(\delta))$ using the assumption made in the corollary, Jensen's inequality, and the fact that by Corollary~\ref{cor:c2} it holds that $\int_\lambda \omega_\lambda d\lambda \geq \omega-\poly(\delta)$. This concludes the proof.
\end{proof}

\subsection{Extension to projection games}
\label{sec:projection}

Theorem~\ref{thm:main} applies to almost consistent symmetric strategies. In this section we give an example of how the results of the theorem can be applied to a family of games such that success in the game naturally implies a bound on consistency. This partially extends the main result in~\cite{manvcinska2014maximally}, with the caveat that our result applies only to projection games, and not the more general ``weak projection games'' considered in~\cite{manvcinska2014maximally}; it is not hard to see that this is necessary to obtain a ``robust'' result of the kind we obtain here. 

\begin{definition}
A game $\game = (\mX, \mY,\mA,\mB, \nu,D)$ is a \emph{projection game} if for each $(x,y)\in\mX\times\mY$ there is $f_{xy}:\mA\to\mB$ such that $D(a,b|x,y)=0$ if $b\neq f_{xy}(a)$. 
\end{definition}

\begin{theorem}\label{thm:proj}
There are universal constants $c,C>0$ such that the following holds. Let $\game = (\mX,\mY,\mA,\mB,\nu,D)$ be a projection game and $\strategy = (\ket{\psi},A,B)$ a  strategy for $\game$ that succeeds with probability $1-\eps$, for some $0\leq \eps\leq 1$. 
 Then there is a measure $\mu$ on $\mathbb{R}_+$ and a family of Hilbert spaces $\mH_\lambda \subseteq \mH_\reg{A}$, for $\lambda\in\mathbb{R}_+$ (both depending on $\ket{\psi}$ only) such that the following holds. 
For every $\lambda\in \R_+$ there is a PME strategy $\strategy_\lambda=(\ket{\psi_\lambda},A^\lambda,B)$ for $\game$ such that $\ket{\psi_\lambda}$ is a maximally entangled state on $\mH_\lambda \otimes \mH_\lambda$ and moreover if $\omega_\lambda$ is the success probability of $\strategy_\lambda$ in $\game$ then 
\begin{equation}\label{eq:lambda-strat-0c}
 \int_\lambda \omega_\lambda d\mu(\lambda) \,\geq\,  1-C\,\eps^c\;.
\end{equation}
\end{theorem}

\begin{proof}
Applying Naimark's theorem (Lemma~\ref{lem:naimark}), extending $\ket{\psi}$ if necessary we may assume that for every $x,y$, $\{A^x_a\}$ and $\{B^y_b\}$ are projective measurements. 
For each $x\in\mX$ and $a\in\mA$ let 
\[ B^x_a = \Es{y\sim\nu_x} \sum_{b} D(a,b|x,y) B^y_b\;,\]
where for $x\in\mX$, $\nu_x$ is the conditional distribution of $\nu$ on $\mY$, conditioned on $x$.  The assumption that $\game$ is a projection game implies that for every $x$, $\sum_a B^x_a \leq \Id$. Let $\ket{\psi_A}$ and $\ket{\psi_B}$ be the canonical purifications of the reduced density of $\ket{\psi}$ on $\mH_\reg{A}$ and $\mH_\reg{B}$ respectively. 
Using Lemma~\ref{lem:psi-cs}, 
\begin{align*}
1-\eps &= \Es{x} \sum_a \bra{\psi}A^x_a \otimes B^x_a \ket{\psi}\\
&\leq \Big( \Es{x}\sum_a \bra{\psi_A} A^x_a \otimes (A^x_a)^T \ket{\psi_A}  \big) \Big)^{1/2} \Big( \Es{x}\sum_a\bra{\psi_B} B^x_a \otimes (B^x_a)^T \ket{\psi_B} \big) \Big)^{1/2} \;,
\end{align*}
which implies that 
\begin{equation}\label{eq:acons-1}
 \Es{x}\sum_a \bra{\psi_A} A^x_a \otimes  (A^x_a)^T \ket{\psi_A} \,\geq\, (1-\eps)^2 \geq 1-2\eps\;.
\end{equation}
Eq.~\eqref{eq:acons-1} shows that the symmetric projective strategy $\strategy_A = (\ket{\psi_A},A)$ satisfies $\delta_\sync(\strategy_A,\nu)\leq 2\eps$. Thus we can apply Theorem~\ref{thm:main}. Let $\mu$, $\mH_\lambda \subseteq \mH_\reg{A}$, $A^\lambda$ be as promised by the theorem. Since $\dim(\mH_{\reg{B}})\geq d$ for each $\lambda$ we can find a purification $\ket{\psi_\lambda}_{\reg{AB}}$ of $\rho_\lambda = P_\lambda/\Tr(P_\lambda)$ on $\mH_\reg{A}\otimes \mH_\reg{B}$; note that  $\ket{\psi_\lambda}_{\reg{AB}}$ is maximally entangled. 

Next we note that 
\begin{align}
\Big|\Es{x}\sum_a \Tr\big(A^x_a B^x_a \rho\big) &- \int_\lambda \Es{x}\sum_a \Tr\big(A^{\lambda,x}_a B^x_a \rho_\lambda\big)d\mu(\lambda)\Big|\notag\\
&\leq \Big( \int_\lambda \Es{x}\sum_a \Tr\big( (A^x_a - A^{\lambda,x}_a)^2 \rho_\lambda d\mu(\lambda)\Big)^{1/2} \Big( \int_\lambda \Es{x}\sum_a \Tr\big( (B^x_a)^2 \rho_\lambda d\mu(\lambda)\Big)^{1/2} \notag\\
&\leq \poly(\eps)\;,\label{eq:close-1}
\end{align}
where the first inequality is Cauchy-Schwarz and uses~\eqref{eq:rho-plambda} and the second uses~\eqref{eq:acons-1} to bound the first term by~\eqref{eq:lambda-strat-0} and that for all $x$, $\sum_a B^x_a \leq \Id$ to bound the second by $1$. Using that $A^{\lambda,x}_a \rho_\lambda = \rho_\lambda A^{\lambda,x}_a$ since $\rho_\lambda$ is totally mixed and $\{A^{\lambda,x}_a\}$ is supported on it we have that 
\begin{align}
\int_\lambda \Es{x}\sum_a \Tr\big(A^{\lambda,x}_a B^x_a \rho_\lambda\big)d\mu(\lambda)
&= \int_\lambda \Es{x}\sum_a \frac{1}{\Tr(P_\lambda)}\Tr\big(A^{\lambda,x}_a P_\lambda B^x_a P_\lambda\big)d\mu(\lambda)\notag\\
&= \int_\lambda \Es{x}\sum_a \bra{\psi_\lambda} A^{\lambda,x}_a \otimes B^x_a \ket{\psi_\lambda}d\mu(\lambda)\;.\label{eq:close-2}
\end{align}
Eq.~\eqref{eq:close-1} and~\eqref{eq:close-2} together give~\eqref{eq:lambda-strat-0c}. 
\end{proof}

\paragraph{Acknowledgments.} I thank Laura Man{\v{c}}inska, William Slofstra and Henry Yuen for comments and Vern Paulsen for pointing out typos in an earlier version. I thank Junqiao Lin for pointing out a mistake in the proof of Corollary 3.3 in an earlier version. This work is supported by NSF CAREER Grant CCF-1553477, AFOSR YIP award number FA9550-16-1-0495, MURI Grant FA9550-18-1-0161 and the IQIM, an NSF Physics Frontiers Center (NSF Grant PHY-1125565) with support of the Gordon and Betty Moore Foundation (GBMF-12500028). 

\paragraph{Data availability.} No new data were created or analyzed in this 
study.

\bibliography{me}


\end{document}